\documentclass[12pt]{article}

\usepackage{float}
\usepackage{graphics, graphicx}
\usepackage{amsfonts}
\usepackage{ natbib}

\usepackage[margin=10pt,font=small,labelfont=bf,labelsep=period]{caption}

\usepackage{amsmath,amsbsy,amssymb,amsthm,latexsym, tabularx}
\usepackage{multirow,multicol,subfigure,booktabs,url}

\newtheorem{thm}{Theorem}
  
\newtheorem*{rem}{Remark}

\newcommand{\bt}{\begin{itemize}}
\newcommand{\et}{\end{itemize}}
\newcommand{\ben}{\begin{enumerate}}
\newcommand{\een}{\end{enumerate}}
\newcommand{\beq}{\begin{equation}}
\newcommand{\eeq}{\end{equation}}
\newcommand{\beqn}{\begin{eqnarray}}
\newcommand{\eeqn}{\end{eqnarray}}
\newcommand{\bea}{\begin{eqnarray*}}
\newcommand{\eea}{\end{eqnarray*}}
\newcommand{\bs}{\boldsymbol}

\newcommand{\inprob}{\stackrel{\mathrm{p}}{\longrightarrow}}
\newcommand{\dlaw}{\stackrel{\mathrm{\mathcal{L}}}{\longrightarrow}}
\newcommand{\iid}{\stackrel{\mathrm{iid}}{\sim}}

\usepackage{color,hyperref}
\definecolor{darkblue}{rgb}{0.0,0.0,0.55}
\hypersetup{colorlinks,breaklinks,linkcolor=darkblue,urlcolor=darkblue, anchorcolor=darkblue,citecolor=darkblue}

\renewenvironment{itemize}{
\begin{list}{}{
\setlength{\leftmargin}{4em}
}}{
  \end{list} }

\begin{document}

\title{A Generalized Kruskal--Wallis Test Incorporating Group Uncertainty with Application to Genetic Association Studies}

\author
{Elif F. Acar$^{1,2}$ and
Lei Sun$^{3,1}$\\
\small{$^{1}$ Department of Statistics, University of Toronto}\\
\small{$^{2}$ Department of Mathematics and Statistics, McGill University (current affiliation)}\\
\small{$^{3}$ Dalla Lana School of Public Health, University of Toronto}
}
\date{}

\maketitle

\begin{abstract}
Motivated by genetic association studies of SNPs with genotype uncertainty, we propose a generalization of the Kruskal--Wallis test that incorporates group uncertainty when comparing $k$ samples. 
The extended test statistic is based on probability-weighted rank-sums and follows an asymptotic chi-square distribution with $k-1$ degrees of freedom under the null hypothesis. 
Simulation studies confirm the validity and robustness of the proposed test in finite samples. 
Application to a genome-wide association study of type 1 diabetic complications further demonstrates the utilities of this generalized Kruskal--Wallis test for studies with group uncertainty.
\end{abstract}

{\bf Keywords:} {\em Genome-wide association studies; Imputation; k-sample problem; Misclassification; Next-generation sequencing; Non-parametric test; Probabilistic data; Rank.}

\newpage
\section{Introduction}
The seminal work by Kruskal and Wallis (\citeyear{KW:1952}) provided us a robust rank-based test for the $k$-sample problem, complementing the parametric approaches such as the one-way analysis of variance (ANOVA). In the classical  $k$-sample problem, data are well classified into different categories or groups.  
However, in many current scientific studies the categorical variables are not necessarily deterministic, and the uncertainties are quantitatively expressed by probability distributions over attributes. Such classification problems often arise in biomedical and bioinformatics applications where data mining techniques and classification algorithms are used to obtain class membership probabilities.

A particular motivating example of this work is the genetic association study of single nucleotide polymorphisms (SNPs) for which the genotype group assignments are often known with ambiguity. The data uncertainties at these SNPs are typically represented by genotype probabilities obtained from various genotype calling algorithms \citep[e.g.][]{Carvalho:2010} or imputation algorithms \citep[e.g.][]{Mach:2009}. 
Table~\ref{tab1} provides a hypothetical illustration.
In such cases, a number of parametric remedies have been proposed, including the popular dosage approach, the weighted regression method \citep{ProbABEL} and likelihood-based score tests \citep{Schaid:2002}.
Although these parametric approaches are satisfactory in many applications, investigators often seek complimentary evidence provided by robust non-parametric  alternatives, safeguarding their statistical analyses against potential model misspecifications. Therefore, it is of both theoretical and practical importance to generalize the Kruskal--Wallis test  so that it is applicable to the $k$-sample problem but with group uncertainty.

\begin{table}[h!] \small
\begin{center} 
\caption{An illustration of SNP genotype probabilities \label{tab1}}
\begin{tabular}{c cc     cccccccc    c c } \label{imputed}
\\[0.5ex] 
\hline
&   Individual    && \multicolumn{6}{c} {Genotype}     & hard call &  soft call  
\\[0.5ex]
&                         &&0    &&     1       &&     2           &&  &
\\  \hline
&1   &&   0.925  &&   0.045   &&  0.030 & &    0 & 0.105
\\ [0.5ex]
&2   &&   0.156  &&   0.102   &&  0.742 & &    2 & 1.586
\\ [0.5ex]
& \vdots && \vdots && \vdots && \vdots  &
\\ [0.5ex]
& N  &&   0.375  &&   0.410   &&  0.215 &&  1 &  0.840
\\[0.5ex]
\hline
\end{tabular}
\end{center} 
\end{table}

To formulate the testing problem, let $Y$ be a continuous response variable and $G$ a categorical variable with $k$ distinct attributes. 
For instance, in genetic association studies, $Y$ denotes the phenotype of interest (e.g. blood pressure or glucose level) and $G$ is the genotype variable at a particular SNP with $k=3$. The three categories for a SNP represent if an individual's genotype at this SNP contain $0$, $1$ or $2$ copies of the minor allele (one of the two alleles with population frequency $<0.5$). The goal of the association analysis between $Y$ and $G$ is to determine if the phenotype $Y$ values differ between individuals with different genotype $G$ values.  
This can be achieved, for example, by regressing $Y$ on $G$ and other relevant covariates. Hence, the classic linear model 
\beq
\label{LM}
Y_j = \mu + \beta \;G_j + \varepsilon_j, \qquad \qquad j=1,2,\ldots, N,
\eeq
with $\varepsilon_j \iid N(0, \sigma^2)$, provides a basis for most association analysis or group comparisons. 
However, results from the robust Kruskal--Wallis test are often obtained as preliminary or complimentary evidence. 

In practice, what is available to us may not be the true $G$ group value, but rather probabilistic data of $G$, i.e. the vector of group probabilities $\bs{p}_{j} = (p_{1j}, p_{2j},\cdots, p_{kj} )$, where 
$
p_{ij} = P(G_j = i)
$ 
for $i= 1,2, \ldots, k$ and $j=1,2, \ldots, N$, with $\sum_{i=1}^k p_{ij}=1$. 
In genetic association studies of SNPs, depending on the experiment used by each application, the genotype of a SNP may be inferred via classification algorithms (e.g. Birdseed, \citep{Birdseed}), imputed via imputation algorithms, (e.g. TUNA \citep{Nicolae:2006}, MaCH \citep{Mach:2006} and Impute \citep{Marchini:2007}), or derived from next generation sequencing calling algorithms (e.g.  SYZYGY \citep{SYZYGY} and SNVer \citep{SNVer}).
In each case, an individual's genotype is most likely associated with some level of uncertainty and the probability that the true genotype belongs to each of the three genotype groups, $p_{ij} = P(G_j = i)$ for $i= 0, 1, 2$ is provided for individual $j$, as illustrated in Table~\ref{tab1}.

In the presence of genotype group uncertainty, the best-guess approach (also known as the hard-call approach) bypasses the problem of incomplete group information by using the most probable group, $\tilde{G}_j= \{ i :  p_{ij}= \max(p_{0j}, p_{1j}, p_{2j}) \}$, in place of $G_j$ in \eqref{LM}.
Depending on whether $\tilde{G}_j$ are considered ordinal or categorical, a $t$-test ($df=1$) or an ANOVA $F$-test ($df=2$) can be performed, so is the Kruskal--Wallis test ($df=2$). 
We refer to these tests as Best-Guess Linear Model (BG-LM), Best-Guess ANOVA (BG-ANOVA) and Best-Guess  Kruskal--Wallis  (BG-KW), respectively (Table~\ref{test_table}). 
Although convenient, this best-guess approach fails to fully utilize the information in the group probabilities.
 
An alternative method is the dosage approach (also known as the expectation-substitution or soft-call approach). 
In this approach, each $G_j$ in \eqref{LM} is substituted by its expectation $\bar{G}_j= p_{1j}+2 \times p_{2j}$. 
The association evidence is then assessed, for example, by a $t$-test ($df=1$) from the regression of $Y$ on $\bar{G}$.  
We refer to this test as the dosage test. 
The main disadvantage of this approach is that it does not allow $G$ to be of categorical nature, and it constrains the relationship between $Y$ and $G$ to an additive model. There are several other model-based methods that incorporate group probabilities \citep[e.g.][]{Marchini:2007, Lin:2008, Kutalik:2010, ProbABEL, Schaid:2002}, however, model-free counterparts such as the Kruskal--Wallis test has not been proposed.

The remaining paper is organized as follows.
In Section \ref{s:sec2}, we describe the construction of the generalized Kruskal--Wallis test statistic based on intuitive probability-weighted rank-sums, 
and provide theoretical justifications of its asymptotic chi-square distribution. 
We show that the original Kruskal--Wallis test is a special case of the proposed test and discuss how to handle tied observations and possible variations in probabilistic data. 
Section \ref{s:sec3} contains simulation studies to evaluate the finite sample performance of the proposed test at different levels of group uncertainty. 
Section \ref{s:sec4} applies the method to data from a genome-wide association study of complications in type 1 diabetic patients. 
Section \ref{s:sec5} concludes with additional discussions. 
Additional simulation results are provided in the Supplemental Material.

\section{The Generalized Kruskal--Wallis Test}
\label{s:sec2}
\subsection{Notation and the Original Kruskal--Wallis Test}

Consider a random sample of size $N$ from a large population consisting of $k \geq 2$ disjoint groups or categories, with population proportions $\pi_i $, $i=1,\ldots,k$, and $\sum_{i=1}^k \pi_i =1$. 
Denote by $G$ the categorical variable taking values on $\mathcal{G}=\{1,2,\ldots,k\}$ and suppose that each category is adequately represented in the sample.
Of interest is to compare the $k$ groups, of sizes $ n_1, \ldots, n_k$ with $\sum_{i=1}^k n_i = N$, based on a continuous response variable $Y$. Formally, letting the distribution function of $Y$ over the group $i$ be of the form $F_i(y)= F(y-\theta_i)$, we wish to test 
\beq 
\label{null}
H_0 : \theta_1 = \theta_2 = \cdots = \theta_k  \qquad \text{against} \qquad  H_A: \text{not all $\theta_i$'s are equal.} 
\eeq

When the precise category assignment of $G$ is available, the Kruskal--Wallis test for \eqref{null} is performed by ranking all the observations together and comparing the sum of the ranks for each group. 
Let $r_{j}$ be the rank of $Y_j$ in the overall sample and define the indicator variable 
$
Z_{ij} = \bs{1}(G_j= i)
$
for the group membership, the Kruskal--Wallis test statistic is 
\beq  
\label{H}
H =  \frac{12}{N(N+1) }    \sum_{i=1}^k \frac{R_i^2}{n_i}  - 3(N+1),    
\eeq
where $n_i = \sum_{j=1}^N \; Z_{ij}$ and 
$$R_i = \sum_{j=1}^N \; Z_{ij} r_j.$$

Under the null hypothesis of \eqref{null}, $H$ follows an asymptotic chi-square distribution with $k-1$ degrees of freedom \citep{K:1952,KW:1952}.

\subsection{The generalized  Kruskal--Wallis Test}
Suppose available to are not $G$ but probabilistic data of $G$, $\bs{p}_{j} = (p_{1j}, p_{2j},\cdots, p_{kj} )$, where 
$
p_{ij} = P(G_j = i)
$ 
for $i= 1,2, \ldots, k$ and $j=1,2, \ldots, N$, with $\sum_{i=1}^k p_{ij}=1$. 
In this case, it is intuitive to consider the weighted rank-sum 
\beq
\label{Rtilde}
R^{\ast}_i = \sum_{j=1}^N \; p_{ij} r_j
\eeq
for each group as a basis of comparison.

However, direct replacement of $R_i$ with $R^{\ast}_i$ in the original Kruskal--Wallis test statistic \eqref{H} does not lead to a properly calibrated test statistic. 
Below we describe the construction of the generalized  Kruskal--Wallis test, based on this appealing weighted rank-sum $R^{\ast}_i$,  that has an asymptotic chis-square distribution with $k-1$ degrees of freedom.

We make the following assumptions.
\bt
\item[(A1)] $\bs{p}_{j}$'s are independent of $Y_j$'s.

\item[(A2)]  $\bar{p}_{i}= \sum_{j=1}^N p_{ij} / N \;  \to \pi_i \; $,  as $\; N \to \infty$,  with $ 0< \pi_i<1$.

\item[(A3)]  $\sum_{j=1}^N (p_{ij} -\bar{p}_i )^2  / N  \; \to \nu_i \; > 0$, as $\; N \to \infty$.
\item[(A4)]  $$   \frac{  \sum_{j=1}^N   (p_{ij} -\bar{p}_i ) (p_{i'j} -\bar{p}_{i'} ) }  {\sqrt{  \sum_{j=1}^N (p_{ij} -\bar{p}_i )^2   \;  \sum_{j=1}^N (p_{i'j} -\bar{p}_{i'} )^2 }     }    \to \rho_{ii'} \; \;  \text{as}  \; N \to \infty.  \hspace{1.08in} \quad $$
\et\

The independence assumption (A1) is standard in statistical analyses of explanatory/response data and is reasonable in practice because, for instance, most imputation algorithms do not consider the phenotype data in the classification of the genotype variable \citep{Marchini:2010}.
The assumption (A2) ensures that $\sum_{j=1}^N p_{ij}$ provides a reasonable approximation for $n_i$ and that the relative group sizes are convergent. 
Assumptions (A3) and (A4) are required for the (joint) asymptotic normality of the linear rank statistic $R^{\ast}_i$.

If the groups are indeed from an identical population, the $r_j$'s can be viewed as a random sample drawn without replacement from the first $N$ integers. Thus, $\text{E}(r_j)= (N+1)/ 2 $, $\text{Var}(r_j)=(N^2-1)/ 12 $ and $\text{Cov}(r_j,r_{j'})= -(N+1) / 12 $, for $j\neq j'$ under the null hypothesis of \eqref{null}.
The conditional mean and the conditional variance of $R^{\ast}_i$ given the group probabilities $p_{ij}$ can then be derived as
$$
\text{E}(R^{\ast}_i)=\tilde{\mu}_i = \frac{N+1}{2} \sum_{j=1}^N p_{ij} \qquad \text{and} \qquad  \text{Var}(R^{\ast}_i)=\tilde{\sigma}^2_i =\frac{N(N+1)}{12} \sum_{j=1}^N (p_{ij} - \bar{p}_i)^2,    
$$
respectively.
An important distinction between our approach and that of \cite{KW:1952} is that we consider permutations of the first $N$ integers rather than finite-sampling, hence no finite sample correction is required in our derivations.

\subsubsection{The case with two samples}
The following result is due to the asymptotic theory of linear rank statistics \citep{WW:1944, Hajek:1999} and governs our test construction.
\begin{thm} \label{thm1}
Under the null hypothesis of \eqref{null} and assumptions (A1)-(A3), the limiting distribution of 
$$ L_N=  \frac{R^{\ast}_i  -  (N+1)/2 \;  \sum_{j=1}^N p_{ij}      }{\tilde{\sigma}_i} $$
is standard normal with mean $0$ and variance $1$.
\end{thm}

\begin{proof}
It suffices to show that the sequence $( p_{i1},p_{i2}, \ldots, p_{iN})  $ satisfies the W condition of the Wald--Wolfowitz Theorem \citep{WW:1944}, i.e. 
$$
\frac{  N^{-1} \sum_{j=1}^N (p_{ij} - \bar{p}_i)^r } {   \left\{  N^{-1} \sum_{j=1}^N (p_{ij} - \bar{p}_i)^2  \right\}^{r/2}   } = O(1), \qquad r=3,4,\ldots. 
$$
Since $0 \leq p_{ij} \leq1$, with $  \bar{p}_i  \to \pi_i$, the central sample moments of the sequence $( p_{i1},p_{i2}, \ldots, p_{iN})  $ are finite: 
$$
\frac{1}{N} \sum_{j=1}^N \left(  p_{ij}- \bar{p}_i \right)^r = O(1), \quad  \text{for} \;  r= 3,4,\ldots.
$$
Assumption (A3) ensures a non-zero variance for $( p_{i1},p_{i2}, \ldots, p_{iN})$ and hence the W condition holds. 
For a detailed proof of the asymptotic normality of $L_N$ see, for example, Theorem 6.1 of \cite{DAS:1957}.
\end{proof}

The generalized Kruskal--Wallis test for the two-sample problem takes the form
\beqn \label{Hstar_k2}
H^\ast = \left(    R^{\ast}_i - \frac{N+1}{2} \;  \sum_{j=1}^N p_{ij}  \right)^2   \bigg/  \tilde{\sigma}^2_i,
\eeqn
where $i = 1$ or $2$. By Theorem \ref{thm1}, $H^\ast$ has an asymptotic $\chi^2 (1)$ distribution under the null hypothesis.
Note that it is sufficient to consider only one of the $R^\ast_i$'s in $H^\ast$ as 
$$
  \frac{ R^{\ast}_1 - (N+1)/2 \;  \sum_{j=1}^N p_{1j}}{ \tilde{\sigma}_1   } = -\;  \frac{ R^{\ast}_2 -(N+1)/2 \;  \sum_{j=1}^N p_{2j}}{  \tilde{\sigma}_2 },
$$
which can be easily verified using $p_{1j} = 1-p_{2j}$, for $j=1,\ldots,N$.

\subsubsection{The case with three samples}
For the three-sample problem, we shall consider the joint distribution of any two $R^{\ast}_i$ and $R^{\ast}_{i'}$.
The covariance between $R^{\ast}_i$ and $R^{\ast}_{i'}$ can be calculated as
\beqn
\nonumber
\text{Cov}(R^{\ast}_i, R^{\ast}_{i'}) 
&=& \frac{N(N+1)}{12} \sum_{j=1}^N   (p_{ij} -\bar{p}_i ) (p_{i'j} -\bar{p}_{i'} ), 
\eeqn
which yields the correlation 
 $$ \tilde{\rho}_{ii'}=  \frac{  \sum_{j=1}^N   (p_{ij} -\bar{p}_i ) (p_{i'j} -\bar{p}_{i'} ) }  {\sqrt{  \sum_{j=1}^N (p_{ij} -\bar{p}_i )^2   \;  \sum_{j=1}^N (p_{i'j} -\bar{p}_{i'} )^2 }     }\;. $$
\begin{thm} \label{thm2}
Under the null hypothesis of \eqref{null} and assumptions $(A1)-(A4)$, the limiting distribution of 
\beq
\nonumber
L_N=  \frac{R^{\ast}_i  -     (N+1)/2  \;  \sum_{j=1}^N p_{ij}}{\tilde{\sigma}_i}  \qquad \text{and} \qquad  L'_N= \frac{R^{\ast}_{i'}  -    (N+1)/2 \;    \sum_{j=1}^N p_{i'j}     }{\tilde{\sigma}_{i'}} 
\eeq
is bivariate normal with means $0$, variance $1$ and correlation $\rho_{ii'}$.
\end{thm}

\begin{proof}
Since the sequences $( p_{i1},p_{i2}, \ldots, p_{iN})  $ and $( p_{i'1},p_{i'2}, \ldots, p_{i'N})$ satisfy the W condition, the result directly follows from Theorem 6.3 of \cite{DAS:1957}.
\end{proof}
Similar to \cite{KW:1952}, the generalized Kruskal--Wallis test can be formed by considering the exponent of the bivariate normal distribution of Theorem \ref{thm2}, multiplied by -2, 
\beqn \nonumber
&& H^\ast  =  \frac{1}{1- \tilde{\rho}_{ii'}^{ 2} } \Bigg[   \left(  \frac{R^{\ast}_i  -   (N+1)/2 \;  \sum_{j=1}^N p_{ij}    }{\tilde{\sigma}_i}   \right)^2  +    \left(  \frac{R^{\ast}_{i'}  -  (N+1)/2 \;  \sum_{j=1}^N p_{i'j}     } {\tilde{\sigma}_{i'} }   \right)^2  
\\
\label{Hstar}
&& \qquad \quad   -2 \tilde{\rho}_{ii'}   \left(  \frac{R^{\ast}_i  - (N+1)/2  \; \sum_{j=1}^N p_{ij}    }{\tilde{\sigma}_i}   \right) \left(  \frac{R^{\ast}_{i'}  - (N+1)/2 \;  \sum_{j=1}^N p_{i'j}    }{\tilde{\sigma}_{i'}}   \right)     \Bigg],
 \eeqn
which is asymptotically distributed as $\chi^2 (2)$ under the null hypothesis.

An algebraic simplification of $H^\ast$ is difficult because of the $p_{ij}$'s. However, as in the two-sample case, the left-out $R^{\ast}_\ell$, $\ell \neq i,i'$ would not be informative once $R^{\ast}_i$ and $R^{\ast}_{i'}$ are taken into account.

\subsubsection{The case with more than three samples}
Theorem \ref{thm2} states that the joint distribution of any two linear rank statistic is asymptotically bivariate normal. This result together with \emph{the Cramer-Wold device} yields the joint asymptotic multivariate normality of any $k -1$ linear rank statistic.
Thus, the case with $k>3$ is handled by considering the correlation matrix for an arbitrary $k-1$ of the $R_i$'s in a similar fashion. The test statistic is formed by the exponent of the $(k-1)$-variate normal distribution after multiplication by $-2$ and follows an asymptotic $\chi^2 (k-1)$ distribution under the null hypothesis. 
An R-code that implements the generalized Kruskal--Wallis test for any $k \geq 2$ is available at author's website.

\subsection{Further Considerations}

\subsubsection{$H$-test as a special case}
The generalized Kruskal--Wallis test reduces to the original Kruskal--Wallis test when $G_j$ is known for each subject. In this case, we define $p_{ij}= (Z_{1j},Z_{2j}, \ldots, Z_{kj})$. Then, $R^{\ast}_i = R_i$ and it is easy to show that 

\beqn
\nonumber
&& \tilde{\mu}_i  =  \frac{N+1}{2} n_i \; =\; \text{E}(R_i),  \qquad   \quad   
 \tilde{\sigma}^2_i =\frac{n}{12} (N+1)(N-n)   \; =\;   \text{Var}(R_i), 
\\ \nonumber
\text{and that} \hfill
\\ \nonumber
&&  \tilde{\rho}_{ii'} = - \sqrt{ \left( \frac{n_i}{N-n_i}\right)   \left( \frac{n_{i'}}{N-n_{i'}} \right) }  \; = \; \text{Cor}(R_i,R_{i'}), \quad \text{for}\;  i\neq i'.
\eeqn
Hence, $H$-test is a special case of $H^\ast$.

\begin{rem}
A major advantage of the generalized Kruskal--Wallis test is that it does not require a separate treatment for partially uncertain categorial data. 
In line with the above approach, when $G_j$ is available, the ranks are allocated to the corresponding groups with probability $1$, i.e. $p_{ij}= (Z_{1j},Z_{2j}, \ldots, Z_{kj})$ and when there is uncertainty in the group membership, the probability-weighing principle in \eqref{Rtilde} guard the testing procedure against possible misclassifications.
\end{rem}

\subsubsection {Correction for ties}
The generalized Kruskal--Wallis test can be corrected for ties in the same way as the original Kruskal--Wallis test. 
When ties occur in the data, one typically assigns the mean of the tied ranks to each member of the tied group. This specification, while not affecting the mean rank, reduces the population variance by $\sum T/(12N)$ and increases the population covariance by $\sum  T/\{12N(N-1)\}$,  where $T= (t-1)t(t+1)$ for each group of ties, $t$ denotes the number of ties in the group, and the summation is taken over all groups. Thus, in the case of ties, the variance of $R^{\ast}_i$ decreases by
$$ \frac{\sum T}{12(N-1)}    \sum_{j=1}^N (p_{ij}- \bar{p}_i)^2      .$$
Similarly, for $i \neq i'$, the covariance between $R^{\ast}_i$ and $R^{\ast}_{i'}$ is reduced by
$$ \frac{\sum  T}{12(N-1)}   \sum_{j=1}^N (p_{ij}- \bar{p}_i)(p_{i'j}- \bar{p}_{i'}).$$

The corrected generalized Kruskal--Wallis test for ties is hence obtained by substituting the adjusted $\tilde{\sigma}_i$ and $\tilde{\rho}_{ii'}$ in $H^\ast$. 
Note that, in many situations the difference between the generalized Kruskal--Wallis test  and its tie-corrected version would be negligible barring an excessive number of ties.
The R-code provided at author's website accommodates tie-correction in the generalized Kruskal--Wallis test.

\subsubsection{Exact distribution}
The exact null distribution of the Kruskal--Wallis test for three samples, each with up to five observations, is given in \cite{KW:1952}. 
More extensive tables are later provided by \cite{Iman:1975} for up to eight observations in each sample. 
With a moderately large number of observations, the exact probability calculations of the $H$ statistic become cumbersome, even with modern computing power.
In the case of the generalized Kruskal--Wallis test, similar tabulations, even for small samples, seem intractable due to probability weights and remain an open problem. 
Similar to the suggestion of \cite{KW:1952}, we recommend using the chi-square approximation when the $\sum_{j=1}^N p_{ij}$'s are at least five.

\subsubsection{Variation in group probabilities}
The proposed generalized Kruskal--Wallis test is conditional on observed data. 
Therefore, in our proposal, we treat the $\bs{p}_j$'s as fixed quantities that define the underlying probability distribution of the unknown $G_j$'s. 
The treatment is reasonable for genetic applications as there is little variation in the genotype probabilities if the same algorithm is employed more than once provided the input data are the same \citep[e.g.][]{Pei:2008}. 

In an unconditional inference, the variation in the group probabilities due to random sampling needs to be taken into account. 
This amounts to specifying a probability distribution for the $p_{ij}$'s, which may not be feasible in practice.
Here, we briefly discuss the validity of the proposed test when the group probabilities are random but treated as fixed, focusing on the case $k=2$. 
For the distribution of the $p_{ij}$'s, while one may consider, for instance, a mixture of two beta distributions, we prefer to keep our argument as general as possible.

Suppose the vector of probabilities $\{p_{i1}, \ldots, p_{iN} \} $ are drawn independently from a probability distribution with mean and variance satisfying conditions analogous to (A1)-(A3).
Let $\mu^\ast_{i}$ and $\sigma^{\ast2}_{i}$ denote the unconditional mean and variance of $R^{\ast}_i$, respectively.
It is easy to see that 
$$\mu^\ast_{i} = \frac{N+1}{2} \sum_{j=1}^N \text{E}( p_{ij}) = \text{E}(\tilde{\mu}_i), $$ 
and, by variance decomposition, we obtain 
\beqn
\nonumber
\sigma^{\ast 2}_{i}  
&=&     \frac{N(N+1)}{12} \sum_{j=1}^N  E\{(p_{ij}- \bar{p}_i)^2 \} + \left( \frac{N+1}{2}\right)^2 \sum_{j=1}^N \text{Var} (p_{ij})        \\
\nonumber
&=&     \text{E}(\tilde{\sigma}^2_i) + \text{Var} (\tilde{\mu}_i),
\eeqn
where $\tilde{\mu}_i$ and $\tilde{\sigma}^2_i$ remain the conditional mean and the conditional variance of $R^{\ast}_i$.

The statistic in $H^\ast$, normalized according to the conditional quantities, can be written as
\beq 
\label{random}
\frac{R^{\ast}_i  - \tilde{\mu}_i}  { \tilde{\sigma}_i}  \; = \;   \frac{   \sigma^{\ast }_i }{ \tilde{\sigma}_i}                           \left(   \frac{R^{\ast}_i -   \mu^\ast_{i}  }   {\sigma^{\ast }_i }      -  
 \frac{  \tilde{\mu}_i  - \mu^\ast_{i}  }  {\sigma^{\ast }_i }      \right).
\eeq
The first quantity in the parenthesis can be shown to be asymptotically standard normal under certain conditions similar to (A1)-(A3), and the second quantity has an approximate normal distribution with mean zero and variance $  \text{Var}( \tilde{\mu}_i) /  \sigma^{\ast2}_{i}, $
by the central limit theorem. Note that
$$
\text{Cov}(R^{\ast}_i, \tilde{\mu}_i) = \left( \frac{N+1}{2}\right)^2 \sum_{j=1}^N \text{Var} (p_{ij})=  \text{Var}( \tilde{\mu}_i)
$$
and hence
$$ \frac{R^{\ast}_i -    \tilde{\mu}_i   }   {\sigma^{\ast }_i }       \dlaw N\left(0, \;  \frac{  \text{E}(\tilde{\sigma}^2_i)  } {  \sigma^{\ast2}_i } \right).
 $$
Using Slutsky's theorem, we can obtain $\left( \tilde{R_i}  - \tilde{\mu}_i\right) /  \tilde{\sigma}_i  \dlaw N(0, 1)$, provided that  $  \tilde{\sigma}^2_i  \inprob  \text{E}(\tilde{\sigma}^2_i) $. The latter condition is satisfied when $\text{Var}\{ (p_{ij}- \pi_i)^2\}= \tau^2 >0$.

Thus, under certain assumptions, the generalized Kruskal--Wallis test statistic remains valid when group probabilities are random.
The application in Section \ref{s:sec4} provides further evidence of this conclusion.

\section{Simulations}
\label{s:sec3}
Here we evaluate the methods via simulation studies. Specifically, 
(i) we evaluate the validity of the asymptotic null distribution of the generalized Kruskal--Wallis test in finite samples, and
(ii) we compare the finite sample performance of the generalized Kruskal--Wallis test with those of commonly used parametric tests as summarized in Table~\ref{test_table}.
We focus on the dosage approach as it is the most popular method used in practice, and previous work that compared various parametric approaches also recommended its usage \citep{YunLi:2011}.  

\begin{table*}[h!] \footnotesize
\centering
\caption{Summary of the association tests compared} \label{test_table}
\begin{tabular}{ l l    c c c  } 
\\ 
\toprule
Test && df & \multirow{2}{*}{}  Incorporate &   Robust to    \\
        &&& Group Uncertainty & Model Assumptions         \\
\midrule \\[-1.5ex]   
\multicolumn{5}{c}{\underline{I: $(p_{0j}, p_{1j}, p_{2j})$ is used to determine the most likely genotype group}}  \\[1ex]     
{\bf BG-LM   }              & Best-Guess Linear Model            & 1           & No                                      & No     \\  [0.5ex]   
{\bf BG-ANOVA }       & Best-Guess  ANOVA                   &  2          & No                                      & No       \\  [0.5ex]   
{\bf BG-KW  }              & Best-Guess  Kruskal--Wallis      &  2           & No                                       & Yes   \\ [1.5ex]     
 \multicolumn{5}{c}{\underline{II:  $(p_{0j}, p_{1j}, p_{2j})$ is used directly in the test}} \\  [1ex]      
{\bf Dosage }              & Dosage                                          & 1           & Yes                                      & No    \\  [0.5ex]   
{\bf GKW }                   & Generalized Kruskal--Wallis      &  2           & Yes                                       & Yes    \\ 
\bottomrule
\end{tabular}
\end{table*}

\subsection{Simulation Methods} 
As a proof of principle and being consistent with the motivation of this work, we simulated genetic association data. 
Phenotype and SNP genotype data for $n=1,000$ individuals were generated as follows. 

The three genotype groups were coded as $G= 0$, $1$ and $2$ copies of the minor allele of a SNP. The SNP of interest had a minor allele frequency of $20\%$, leading to expected group size of  $(n_0, n_1, n_2)= (640, 320, 40)$ based on the multinomial distribution with parameters $(0.64, 0.32, 0.04)$ under the Hardy--Weinberg equilibrium assumption. We also considered minor allele frequency of 10\% and other values. 

The phenotype data were generated from an additive normal model, favorable to the parametric methods.
Without loss of generality, $Y$ values were simulated from normal distribution with equal mean of $ (2,2,2)$ under the null model, and $(1.75, 2, 2.25)$ under the alternative model for the three genotype groups $G= 0$, $1$ and $2$, respectively, with a common variance, $\sigma^2=1$. The mean values were chosen such that power (at the $\alpha=0.01$ level) to detect association between $Y$ and $G$ is about $95\%$ for minor allele frequency of $20\%$ (and about $70\%$ for minor allele frequency of $10\%$) with the given sample size and without genotype uncertainty.

Other simulating parameter values were also considered with varying minor allele frequencies, type 1 error rates and sample sizes, as well as non-normal or non-additive models (Table~\ref{simsum}). Additional results are provided in the Supplemental Material and conclusions are characteristically similar to the ones reported here. 

Given a true genotype $G$, to simulate the probabilistic genotype data, we used the Dirichlet distribution with scale parameters $a$ for the correct genotype category and $(1-a)/2$ for the other two, where $a=1$, $0.9$, $0.8$, and $0.7$, corresponding to an increasing level of group uncertainty ranging from $0\%$ to $30\%$. Under this Dirichlet simulating model, the proportion of the individuals whose most probable (best-guessed) genotypes are the correct ones is approximately $a$ (Table~\ref{coverage}).

\begin{table*}[h!] \footnotesize
\centering
\caption{The empirical proportion of the individuals whose most probable (best-guessed) genotypes are the correct ones. $a$ is the Dirichlet parameter for the correct genotype category. SNP has a minor allele frequency of 20\%. Results for other frequencies are similar. } \label{coverage}
\begin{tabular}{c  c  c  c  c c  c  c  c} 
\\ 
\toprule
$a$        &&      average       &&        $G = 0 $       &&        $G = 1 $    & &       $G = 2 $ \\
\midrule
1    && 1.00 && 1.00 && 1.00 && 1.00 \\ 
0.9 && 0.93 && 0.92 && 0.96 && 0.91\\
0.8 && 0.83 && 0.82 && 0.85 && 0.86 \\ 
0.7 && 0.74 && 0.74 && 0.74 && 0.80 \\ 
\bottomrule
\end{tabular}
\end{table*}

For each set of probabilistic data, $M=10,000$ experiments under the null model and $5,000$ experiments under the alternative model were conducted by simulating only the response data.  
Application in Section \ref{s:sec4} confirms that methods comparison is not affected by how the probabilistic data were generated.

\subsection{Evaluation of Accuracy}
Table~\ref{alpha} provides the empirical type 1 error rates of the five tests considered. 
The group uncertainty does not seem to alter the accuracy of any of the tests in an obvious way. 
For the proposed GKW test, we also compare the quantiles of the empirical distributions with those of the $\chi^2(2)$ distribution. Figures~\ref{QQplots020} and \ref{QQplots010} indicate that the empirical null distribution coincides with the asymptotic one, which is further supported by the Kolmogorov--Smirnov test with $p$-values $0.279$, $0.595$, $0.628$ and $0.599$ for SNP with minor allele frequency of $20\%$ and $0.174$, $0.377$, $0.375$ and $0.194$ for SNP with minor allele frequency of $10\%$, from the lowest to the highest uncertainty levels.

\begin{table*}[h!] \scriptsize
\centering 
\caption{Empirical type 1 error rates of the five tests at $\alpha=0.01$ under a normal null model, for testing the association of a SNP that has minor allele frequency of $20\%$ or $10\%$. $a$ is the parameter value of  the Dirichlet distribution used to simulate genotype probabilities.} \label{alpha}
\begin{tabular}{l l  c  c c c c  c    ccccc } \\
\toprule
 &&      \multicolumn{4}{c}{\textbf{  $\bs{\text{minor allele frequency}= 0.2}$}}   && \multicolumn{4}{c}{\textbf{  $\bs{\text{minor allele frequency}= 0.1}$}}  \\[0.8ex]
\cmidrule(r){3-6}  \cmidrule(r){8-11} 
uncertainty level  &&       0\%   &     10 \%   &     20 \%   &     30 \%   &&   0\%   &     10 \%   &     20 \%   &     30 \%   \\[0.4ex]
 &&     ($a=1$)   &     ($a=0.9$)   &    ($a=0.8$)   &     ($a=0.7$)    &&     ($a=1$)   &     ($a=0.9$)   &    ($a=0.8$)   &     ($a=0.7$)     \\[0.3ex]
\midrule
BG-LM                  &&    0.0109 & 0.0108 & 0.0090 & 0.0096   &&    0.0093 & 0.0091 & 0.0095 & 0.0105  \\
BG-ANOVA         &&    0.0098 & 0.0106 & 0.0109 & 0.0094   &&    0.0083 & 0.0089 & 0.0103 & 0.0102\\
BG-KW                &&     0.0094 & 0.0091 & 0.0100 & 0.0097   &&     0.0075 & 0.0096 & 0.0098 & 0.0094  \\
Dosage               &&    0.0109 & 0.0105 & 0.0091 & 0.0099    &&    0.0093 & 0.0092 & 0.0095 & 0.0095 \\ 
GKW                    &&     0.0094 & 0.0087 & 0.0091 & 0.0088   &&     0.0075 & 0.0088 & 0.0092 & 0.0095\\
\bottomrule
\end{tabular}
\end{table*}

\subsection{Evaluation of Efficiency}

We examine the empirical relative efficiency of the other tests compared to the proposed generalized Kruskal--Wallis test under the alternative model, using the normal additive data favorable to the model-based tests (Figure \ref{additive}). 
The empirical power of each test was obtained using the corresponding empirical type 1 error threshold reported in Table~\ref{alpha}.
Note that, when the true genotypes are used (i.e. $a=1$ with 0\% group uncertainty), the GKW and the BG-KW are equivalent. This is also true for the dosage test and the BG-LM.

As expected, the dosage has the best power when there is no group uncertainty (Figure~\ref{additive} at $0\%$ uncertainty level), because the data were generated under the best scenario for the dosage test (i.e. phenotype $Y$ was normally distributed with population means, $(1.75, 2, 2.25)$, increasing in an additive manner with respect to the number of copies of the minor allele, $G=(0, 1, 2)$). When the minor allele frequency is $20\%$, the power of the dosage test remains (slightly) higher than the generalized Kruskal-Wallis test even as the uncertainty level increases (Figure~\ref{additive020}). However, this is no longer true for detecting SNPs with minor allele frequency of $10\%$ (Figure~\ref{additive010}). For example, with uncertainty level at $30\%$ ($a=0.7$), the relative efficiency of all other tests including the dosage test is about $60\%$ as compared to the generalized Kruskal-Wallis test.

Moreover, in less favorable models, (e.g.\ Figure~\ref{WebFigure7} for non-normal model), the generalized Kruskal-Wallis test can outperform the others even when there is no genotype uncertainty. Additional simulation results with different minor allele frequencies (e.g.\ 0.05 and 0.3), different type 1 error rate (e.g.\ 0.05 and 0.001) and different model assumptions (e.g.\ non-additive model) as presented in Figures~\ref{WebFigure4}-\ref{WebFigure8} all confirm the robustness of the generalized Kruskal-Wallis test. Under scenarios favorable to model-based methods, the generalized Kruskal-Wallis test provides comparable power; with model misspecification or increased genotype uncertainty, the generalized Kruskal-Wallis test can outperform the others.

\begin{figure}[h!]
\centering
\subfigure[] 
{\label{additive020}
   \includegraphics[width=0.45\textwidth]{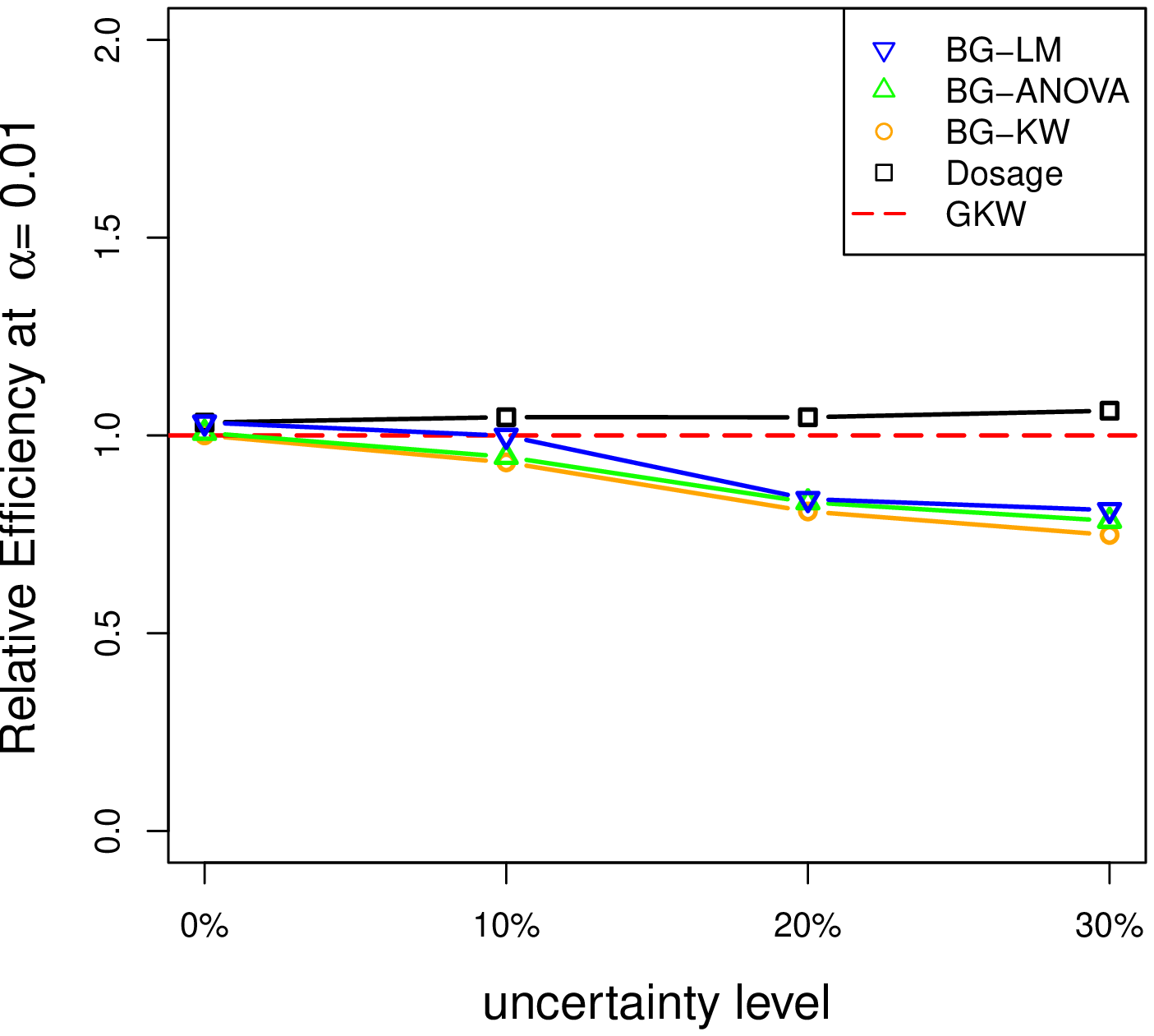}

}
\hspace{0.5cm}
\subfigure[]
{ \label{additive010}
 \includegraphics[width=0.45\textwidth]{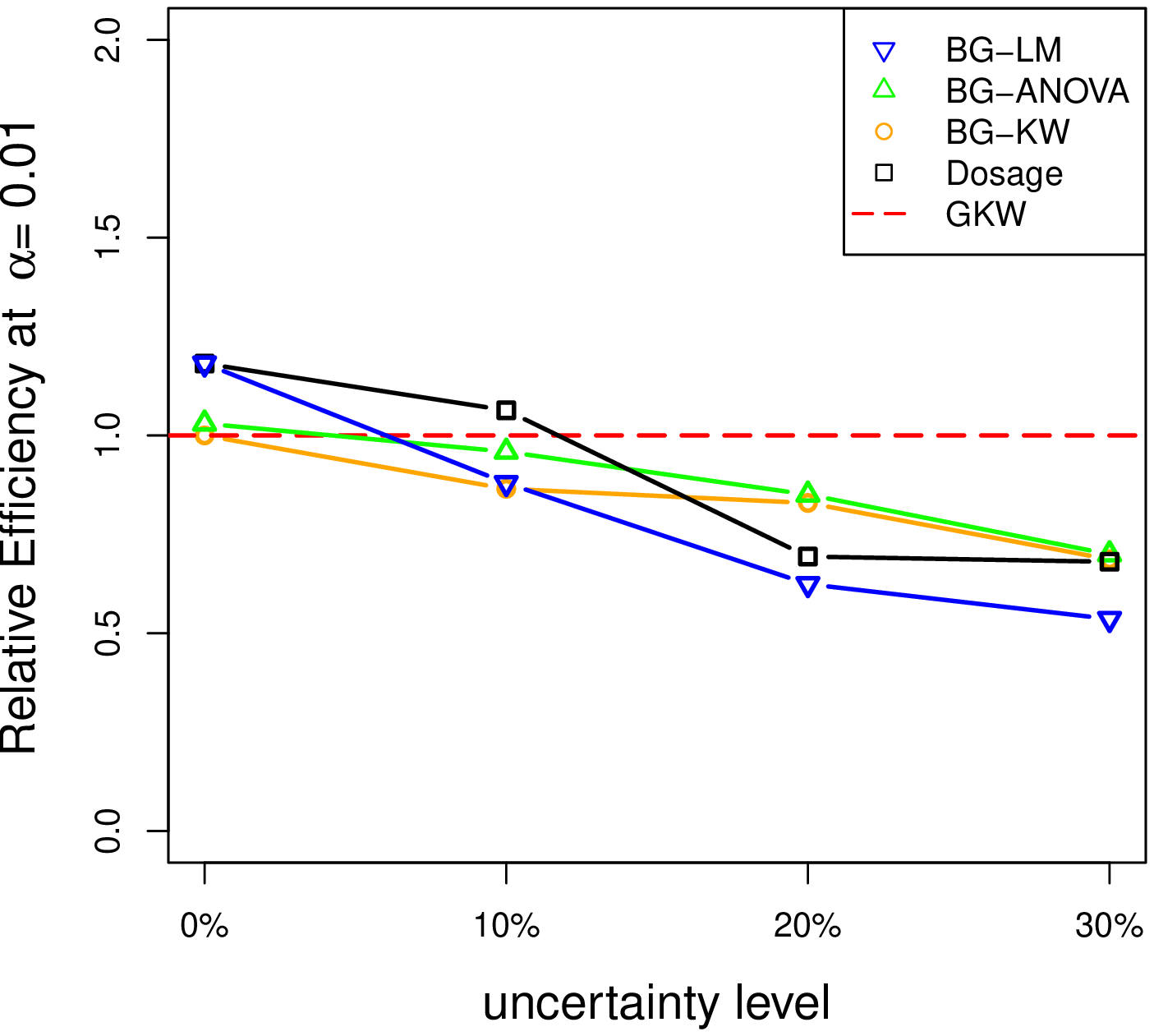}
}
\caption{Relative efficiency of other tests as compared to the GKW test, under a normal additive model, for testing the association of a SNP that has (a) minor allele frequency of $20\%$, or (b) minor allele frequency of $10\%$.    \label{additive}} 
\end{figure}

\section{Data Applications}
\label{s:sec4}
We demonstrate the proposed method in three applications using data from the genome-wide association study of complications in type 1 diabetes patients, for which over 800K SNPs were genotyped by the Illumina 1M beadchip assay and over 1.5 million ungenotyped SNPs were imputed \citep{Paterson:2010}.

The study sample consists of $n=1,300$ subjects with type 1 diabetes ($664$ treated conventionally and $636$ treated intensively) from the Diabetes Control and Complications Trial (DCCT). 
The phenotypes of interest are glycosylated hemoglobin (HbA1c) and diastolic blood pressure (DBP), collected quarterly from each patient over the course of the DCCT. 

The first application illustrates the case of no association using $27,265$ ungenotyped but imputed SNPs on chromosome 22.
The other two applications evaluate the performance of the generalized Kruskal--Wallis test in detecting putative associations of two genotyped SNPs with DBP and HbA1c.

\subsection{$P$-value Distribution of $H^\ast$ on Chromosome 22}

We investigated the $p$-value distribution of the generalized Kruskal--Wallis test using chromosome 22 data under the null hypothesis of no association.
The genotype data at $33,815$ SNPs were imputed (provided by Dr. Andrew Paterson's research group) using MaCH \citep{Mach:2006,Mach:2009} using HapMap II phased data as the reference panel. 
In the association analysis, we considered $27,265$ SNPs that yielded the sum of genotype probabilities of at least five for each of the three genotype groups. 

The phenotype data, mean DBP measurements over the first six study periods, were first permuted to eliminate any possible associations. 
We observed that the null distribution of $p$-values obtained from the $\chi^2(2)$ approximation of the generalized Kruskal--Wallis test statistic closely matches the theoretical uniform distribution.  
Figure~\ref{chr22} displays the corresponding quantile-quantile plot on the $-\log_{10} p$ scale.
The Kolmogorov--Smirnov test for uniformity yields $p$-value$=0.084$, providing additional evidence for the validity of the proposed generalized Kruskal--Wallis test in finite samples, because 
in contrast to the simulations, the group probabilities are random in that their distributions vary between imputed SNPs  and among the patient subjects.

\begin{figure}[h!]
\centering
 \includegraphics[width=0.55\textwidth]{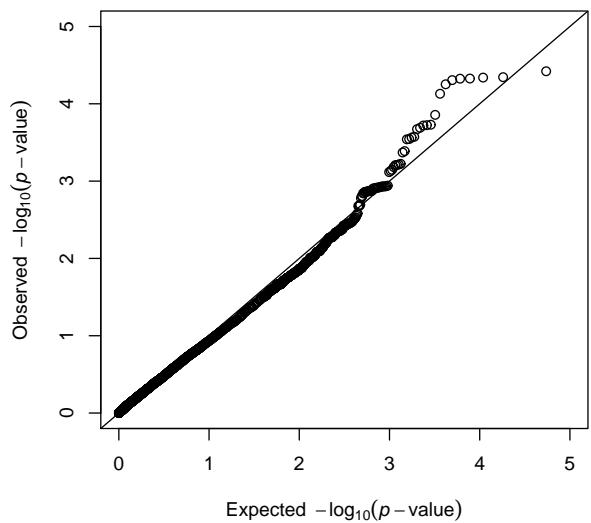}
\caption{The quantile-quantile plot of the  $-\log_{10}(p\text{-values})$ of the GKW test applied to assess association between 27,265 imputed SNPs on chromosome 22 and permuted diastolic blood pressure phenotype measures in 1,300 subjects with type 1 diabetes. \label{chr22}}
\end{figure}

\subsection{DBP with rs7842868 }
Recently, \cite{Ye:2010} identified rs7842868 on chromosome 8 as a SNP associated with DBP with $p\text{-value}\approx 4.5\times10^{-8}$. 
Here, we examined the performance of the five tests in detecting this association.    
Since GKW and BG-KW are readily advantaged in the case of non-normal data, we base our comparisons on the natural logarithm of the DBP measurements averaged over the first six study periods. The histogram of the phenotype data for the $1,300$ patients is given in Figure~\ref{Hist1}.

The observed genotype data at rs7842868 yielded the group sizes $(n_0,n_1,n_2)= (788,446,66)$. 
When tested for association, rs7842868 was found to be significant by all tests with similar results (Figure~\ref{dbp} at 0\% uncertainty level). 
GKW (equivalent to BG-KW when genotypes are known) had a marginally higher statistical significance with $p\text{-value} = 1.36 \times10^{-8}$, followed by the dosage test and BG-LM with $p\text{-value}=2.66\times10^{-8}$.

We then masked the actual data and simulated $1,000$ replicates of {\it in silico} genotypes (i.e. genotype probabilities) from the Dirichlet distribution at each group uncertainty level, as described in Section~\ref{s:sec3}. The results averaged over the $1,000$ replications are shown in Figure~\ref{dbp} and verify the robustness of the GKW test. 
For instance, at $10 \%$ or higher uncertainty level, the proposed GKW test had noticeable better performance than all other procedures.
Nevertheless, power to detect association by any method deteriorated considerably as genotype uncertainty increased.

\begin{figure}[h!]
\centering
\subfigure[]
{
    \label{dbp}
   \includegraphics[width=0.45\textwidth]{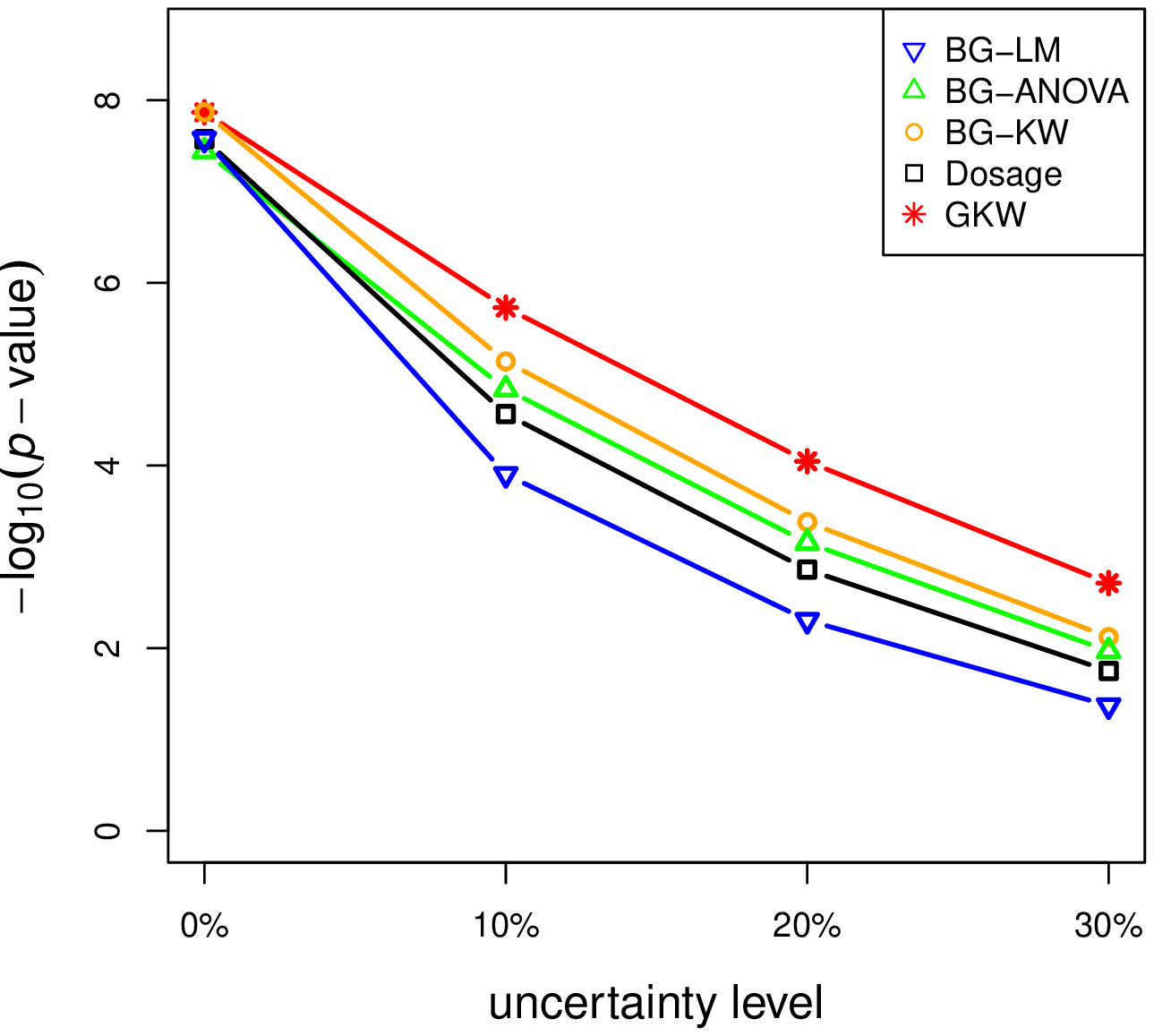}

}
\hspace{0.5cm}
\subfigure[]
{
    \label{hba1c}
   \includegraphics[width=0.45\textwidth]{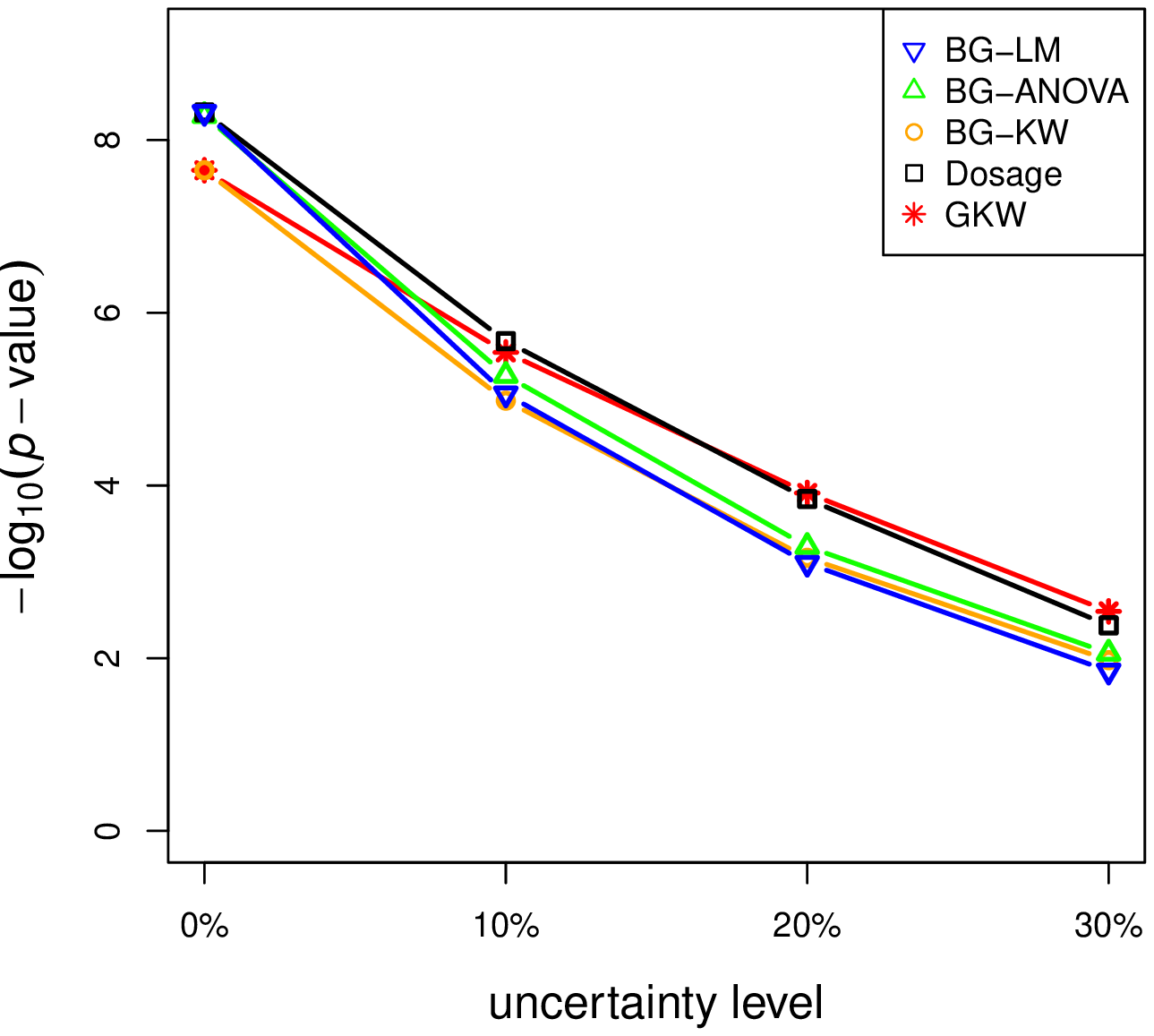} 
}

\caption{The significance,  on the $-\log_{10}(p\text{-values})$ scale, of the association tests at different genotype uncertainty levels for (a) DBP with rs7842868, (b) HbA1c with rs1358030.}
\end{figure}

\subsection{HbA1c with  rs1358030}
The SNP rs1358030 on chromosome 10 was reported to be genome-wide significantly associated with HbA1c ($p\text{-value}=5 \times 10^{-9}$) in the conventionally treated group \citep{Paterson:2010}. 
We performed association analyses using the observed genotypes and masked probabilistic genotypes in a fashion similar to the above DBP application.

The histogram of the phenotype, average $\log(\text{HbA1c})$ measurements over the first six study periods, is given in Figure~\ref{Hist2}. for the $n=664$ conventionally treated patients. 
The genotype group sizes at rs1358030 were $(n_0,n_1,n_2)= (267,307,90)$. 
In this application, when the actual genotypes were used, the dosage test (equivalent to BG-LM) showed the best performance in detecting the HbA1c association (Figure ~\ref{hba1c} at 0\% uncertainty level).
However, the advantage of the dosage test  dissipated as the genotype uncertainty level increased.

\section{Conclusions and Disucssions}
\label{s:sec5}

In this paper, we generalized the rank-based nonparametric  Kruskal--Wallis test to allow for group uncertainty when comparing $k$ samples. 
The proposed generalized test statistic follows an asymptotic chi-square distribution with $k -1$ degrees of freedom, suitable for statistical inference of large-scale data (e.g. genome-wide association or next-generation sequencing data) without the need for permutation or other computationally inefficient procedures. 

Although the work was originally motivated by the analyses of SNPs with genotype uncertainty in genetic association studies, it can be readily applied to other scientific studies. Extensive simulations and several applications showed that the generalized Kruskal--Wallis test provides a good balance between robustness and power. Our proof-of-principle stimulation studies could be improved to more closely mimic real genetic data and models. However, the validity and robustness conclusion characteristically holds, given the combined evidence from our theoretical work, and simulation and application studies.

The proposed generalized  Kruskal--Wallis test has its limitations. For example, the exact distribution for small samples (e.g. relevant to genetic association studies of rare variants) is unknown. The current test does not include other covariates. However, this limitation could be partially alleviated by considering the residuals from a regression model that accounts for the effects of other covariates first, assuming there is no SNP-covariate interaction. 
Both the original and the generalized Kruskal--Wallis tests are formulated for continuous outcomes, however, the probability-weighting principle exploited here could potentially be extended to analyze case-control data or other categorical outcomes and is the subject of ongoing research.

In conclusion, the proposed generalized Kruskal--Wallis test provides scientists a tool to continue investigate, in a robust non-parametric fashion, the $k$-sample problems in the presence of group uncertainty. The generalized Kruskal--Wallis test includes the original Kruskal--Wallis test as a special case, and its power is comparable to parametric counterparts under conditions favorable to the model-based approaches. 
When there is model misspecification or high group uncertainty, the generalized Kruskal--Wallis test can outperform the others. 
 
\section*{Acknowledgement}
The authors thank Dr. Andrew Paterson and his research group, specifically Daryl Waggott and Ye Chang, for providing their genome-wide association data and Drs. Fang Yao and D.A.S. Fraser for constructive discussions. 
The research was supported by Natural Sciences and Engineering Research Council of Canada (NSERC, 250053-2008) and Canadian Institutes of Health Research (CIHR, MOP 84287).

\bibliography{GKW_biblio2}
\bibliographystyle{asa}

\section*{Supplemental Material}

\setcounter{table}{0}
\renewcommand{\thetable}{S\arabic{table}}%
\setcounter{figure}{0}
\renewcommand{\thefigure}{S\arabic{figure}}

This supplement contains the results of additional simulations under various scenarios, as well as the tables and figures cited in the main text. Table~\ref{simsum} summarizes the study purpose and the settings of our simulation studies where sample size is 1,000. Other sample sizes (e.g. 500 or 2,000) were also studied, but results were categorically similar, therefore not reported here.
\vspace{0.1in}

\begin{table*}[h!] \footnotesize
\centering
\caption{Reader Guide for the simulation results.} \label{simsum}
\begin{tabular}{c  c  c  c  c c  c  c  c c c} 
\\ 
\toprule
&   Study the impact of  &  Power &   Normal      &      Additive       &        $\alpha$     &      MAF \\
\midrule

\multicolumn{1}{c}{ \multirow {2}{*}{}}    &     \multicolumn{1}{c}{\multirow {2}{*}{minor allele frequency}}  
   & Figure~1(a)  & yes  & yes & 0.01   & 0.2 \\ 
&& Figure~1(b)  & yes  & yes & 0.01   & 0.1 \\ [2ex]
\midrule             
\multicolumn{1}{c}{ \multirow{10}{*}{}} &  

 \multirow{2}{*}{ minor allele frequency }  
   & Figure~S4(a)  & yes  & yes & 0.01   & 0.05 \\ 
&& Figure~S4(b)  & yes  & yes & 0.01   & 0.3 \\ [2ex]

& \multirow{4}{*}{ type 1 error rate}  
   & Figure~S5(a)   & yes  & yes & 0.05   & 0.2 \\ 
&& Figure~S5(b)  & yes  & yes & 0.001   & 0.2 \\ 
&& Figure~S6(a)   & yes  & yes & 0.05   & 0.1 \\ 
&& Figure~S6(b)   & yes  & yes & 0.001   & 0.1 \\ [2ex]
 
 &\multirow{4}{*}{ model assumptions}  
   & Figure~S7(a)   & no  & yes & 0.01   & 0.2 \\ 
&& Figure~S7(b)  & no  & yes & 0.01   & 0.1 \\ 
&& Figure~S8(a)  & yes  & no & 0.01   & 0.2 \\ 
&& Figure~S8(b)  & yes  & no & 0.01   & 0.1 \\ [2ex]

\bottomrule
\end{tabular}
\end{table*}

\newpage

\clearpage

\begin{figure}[h!]
\centering
 \includegraphics[width=0.8\textwidth]{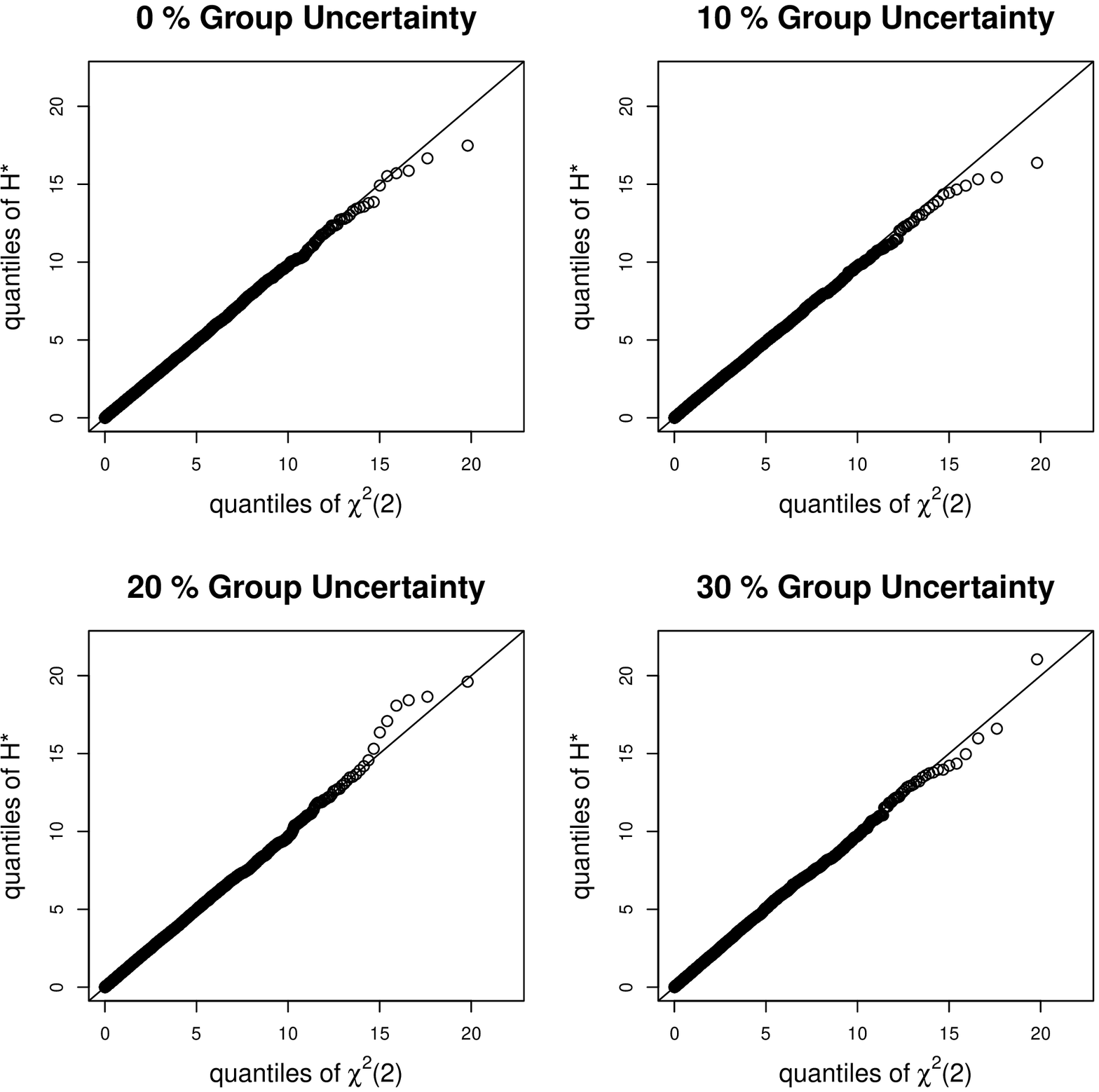}
\caption{The quantile-quantile plots for the GKW test statistic, under a normal null model, for testing the association of a SNP that has minor allele frequency of $20\%$. The Kolmogorov--Smirnov
test has p-values of  $0.279$, $0.595$, $0.628$ and $0.599$, for the lowest to the highest uncertainty levels.  \label{QQplots020} }
\end{figure}

\clearpage

\begin{figure}[h!]
\centering
 \includegraphics[width=0.8\textwidth]{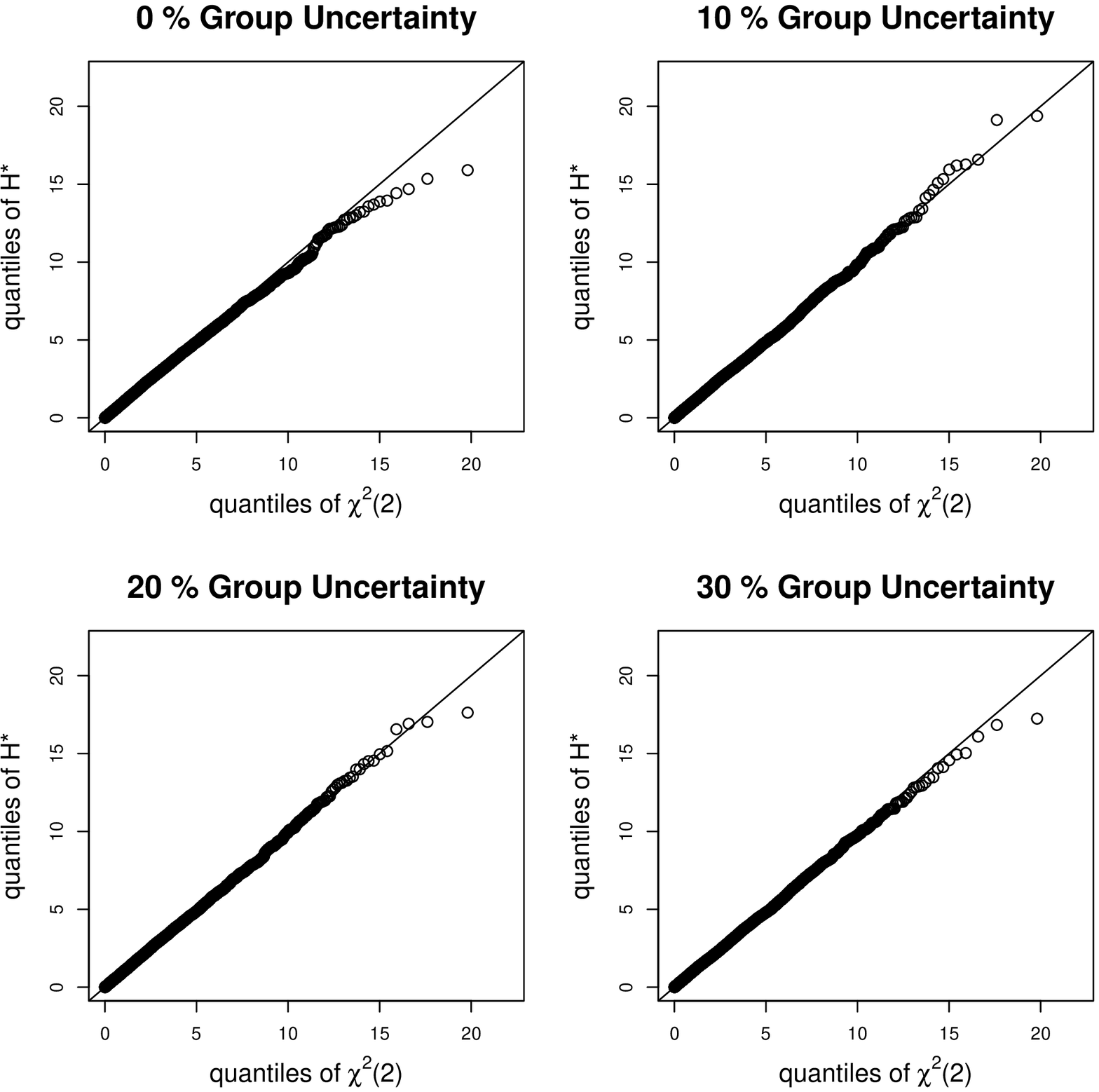}
\caption{The quantile-quantile plots for the GKW test statistic, under a normal null model, for testing the association of a SNP that has minor allele frequency of $10\%$. The Kolmogorov--Smirnov test has p-values of  $0.174$, $0.377$, $0.375$ and $0.194$, for the lowest to the highest uncertainty levels.  \label{QQplots010} }
\end{figure}

\clearpage

\begin{figure}[h!]
\centering
\subfigure[]
{
    \label{Hist1}
   \includegraphics[width=0.5\textwidth]{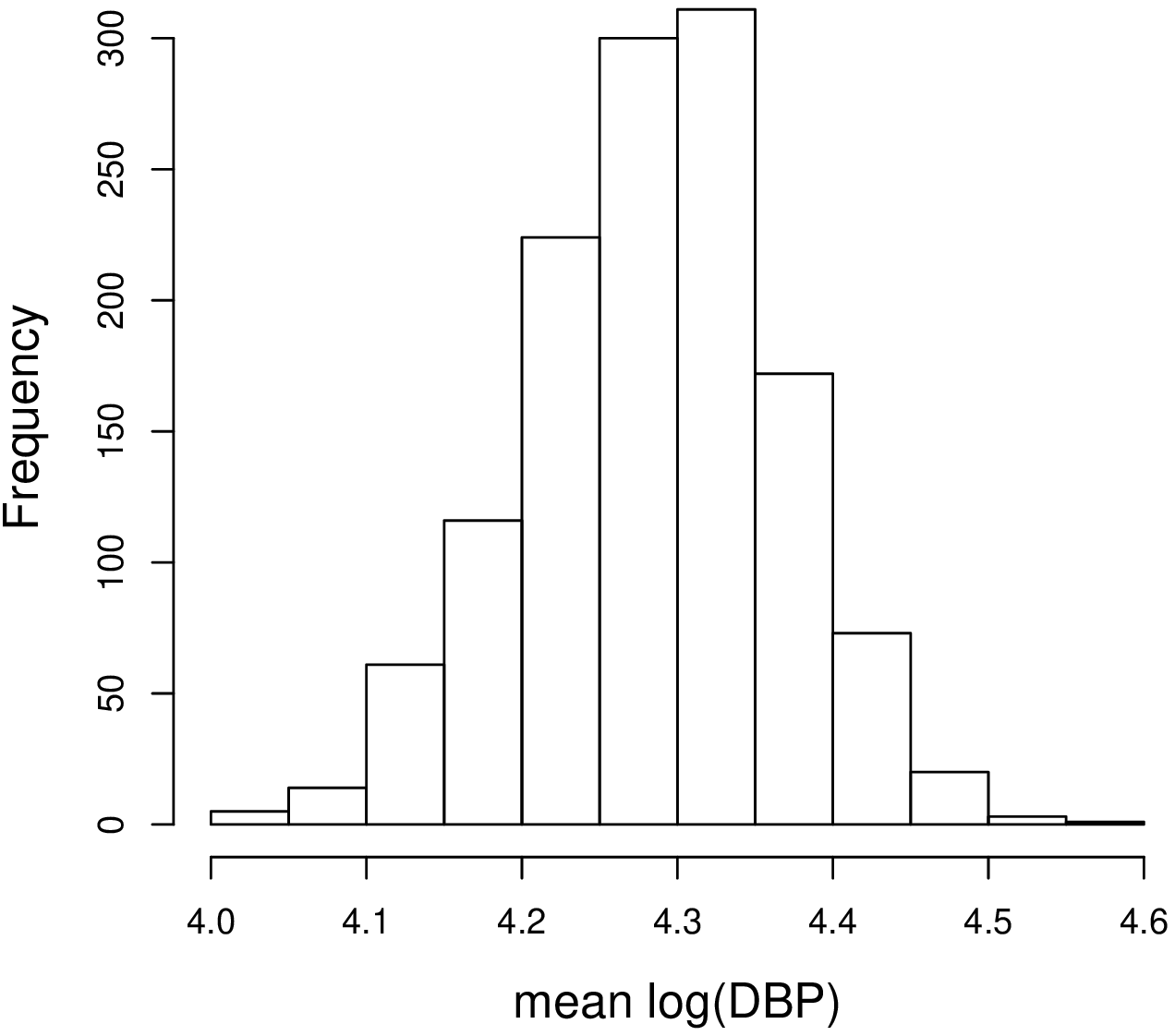}

}
\hspace{0.5cm}
\subfigure[]
{
    \label{Hist2}
   \includegraphics[width=0.5\textwidth]{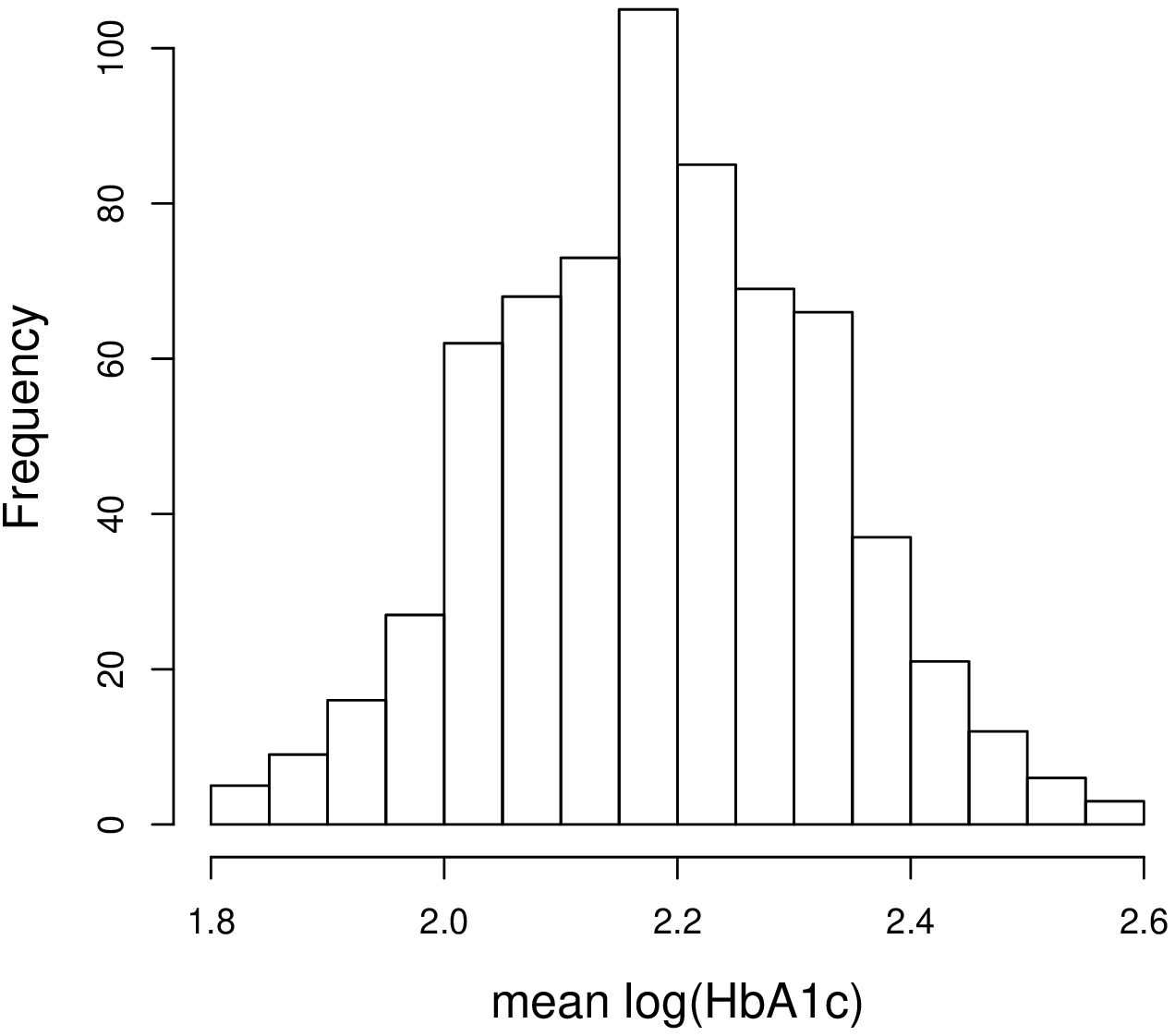} 
}

\caption{Histograms of (a) mean of the natural logarithm of DBP measurements of $1,300$ patients over the first six study periods, (b) mean of the natural logarithm of HbA1c measurements of $664$ conventionally treated patients over the first six study periods. }
\end{figure}

\clearpage

\begin{figure}[h!]
\centering
\subfigure[]
{  \label{WebFigure4a}
   \includegraphics[width=0.5\textwidth]{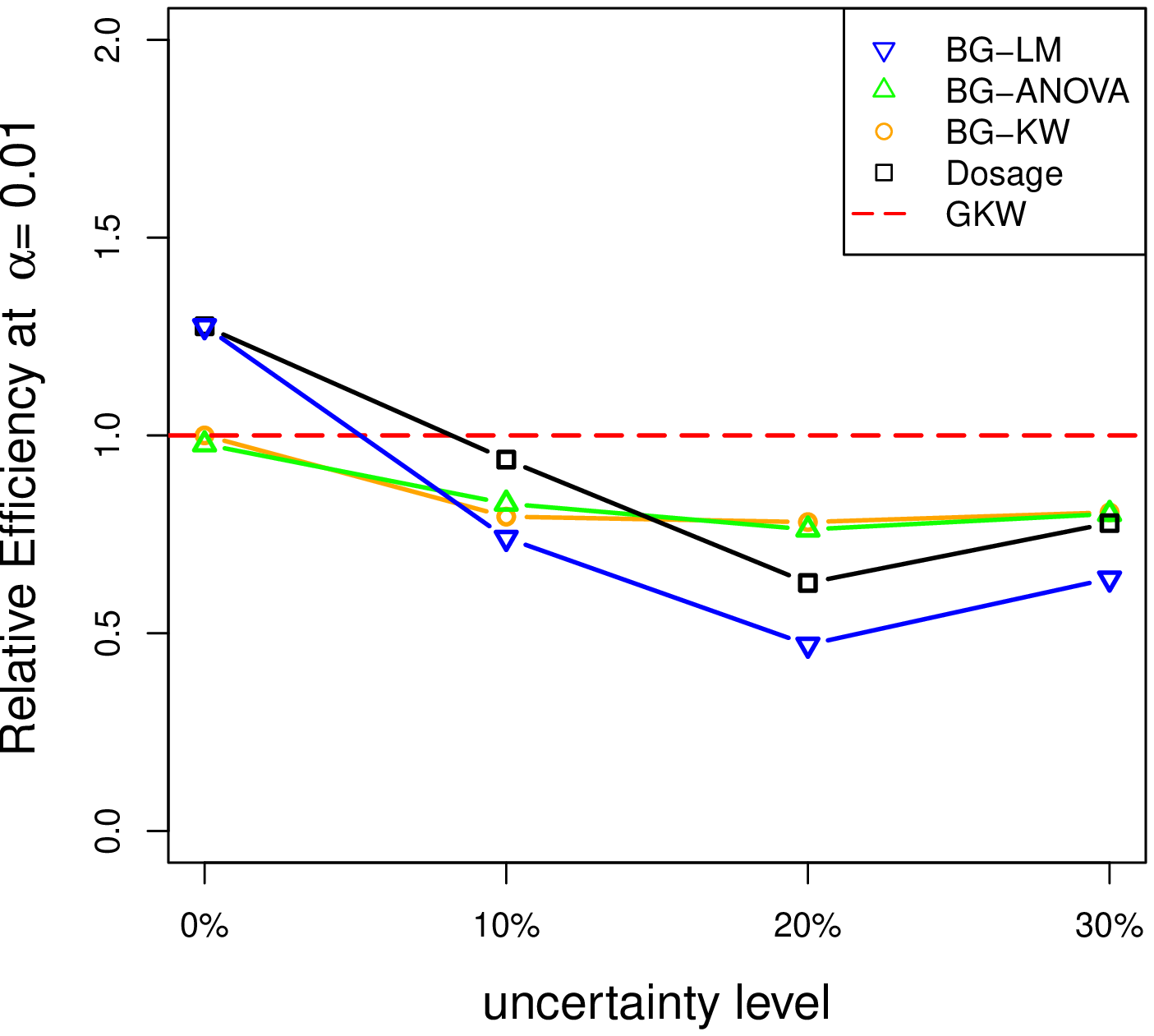}
}
\hspace{0.5cm}
\subfigure[]
{  \label{WebFigure4b}
   \includegraphics[width=0.5\textwidth]{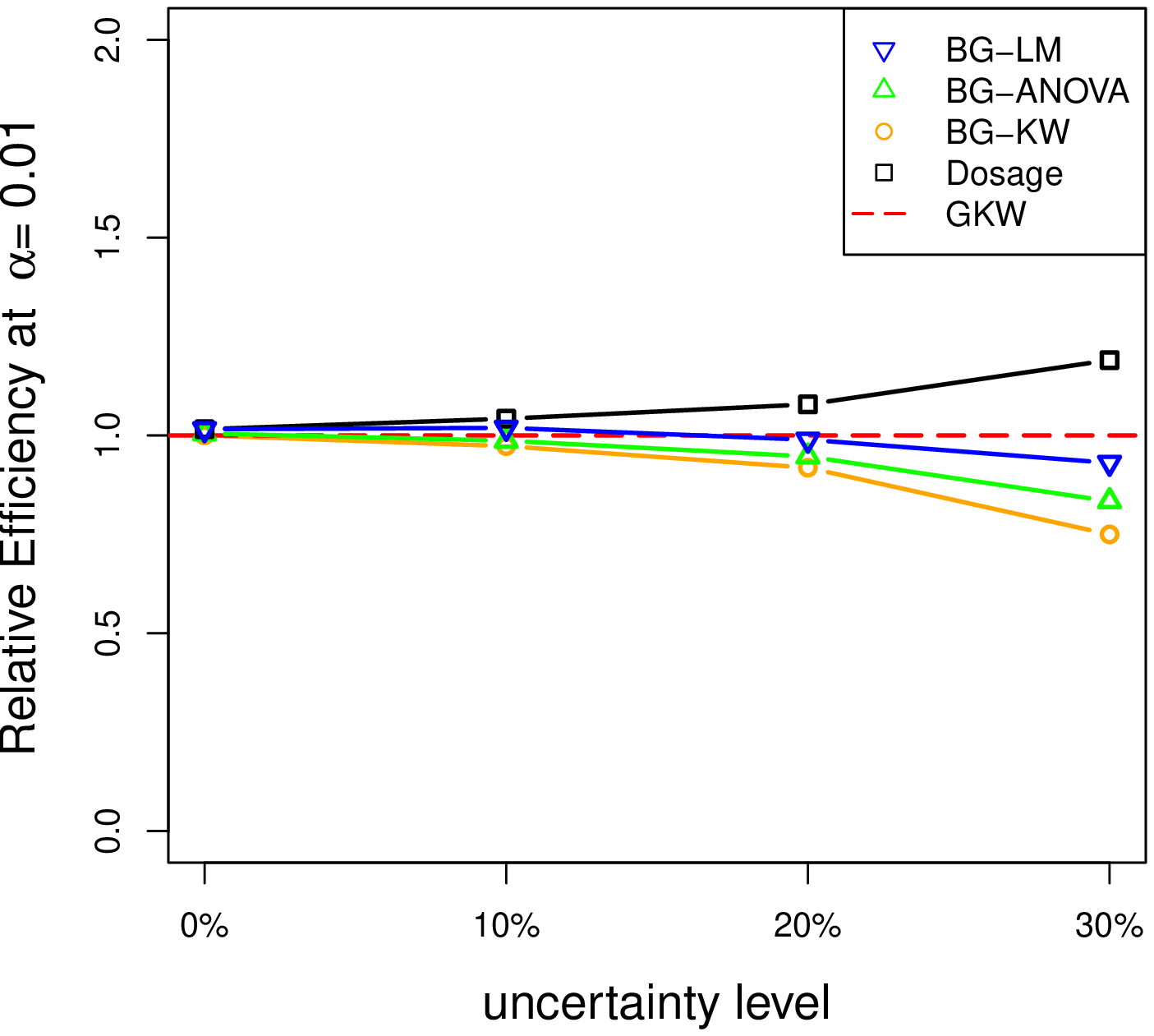} 
}
\caption{Relative efficiency of other tests as compared to the GKW test at $\alpha=0.01$, under a normal additive model, for testing the association of SNP that has (a) minor allele frequency of $5\%$, (b) minor allele frequency of $30\%$. \label{WebFigure4}}
\end{figure}

\clearpage

\begin{figure}[h!]
\centering
\subfigure[]
{ \label{WebFigure5a}
   \includegraphics[width=0.5\textwidth]{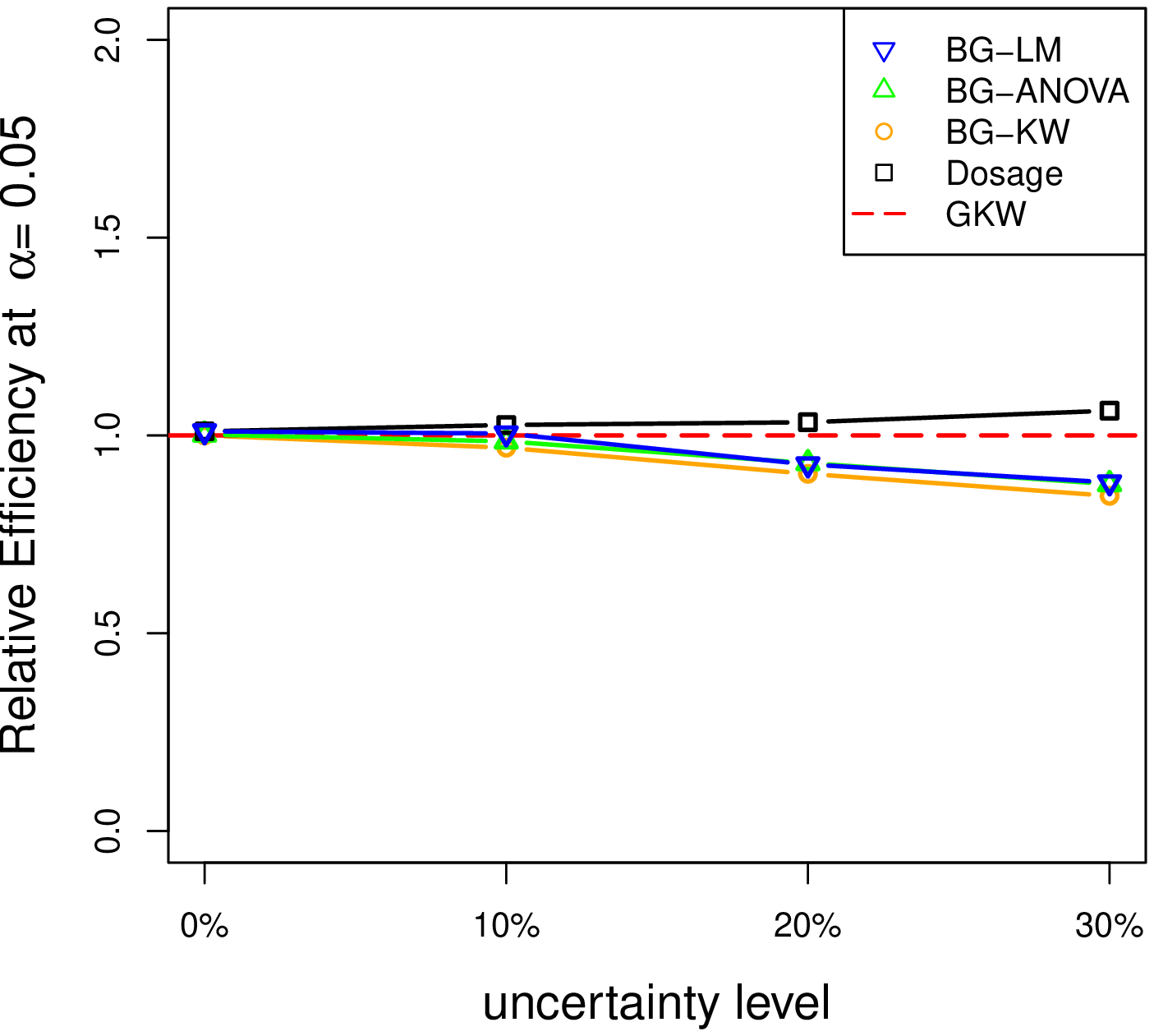}
 }
\hspace{0.5cm}
\subfigure[]
{\label{WebFigure5b}
  \includegraphics[width=0.5\textwidth]{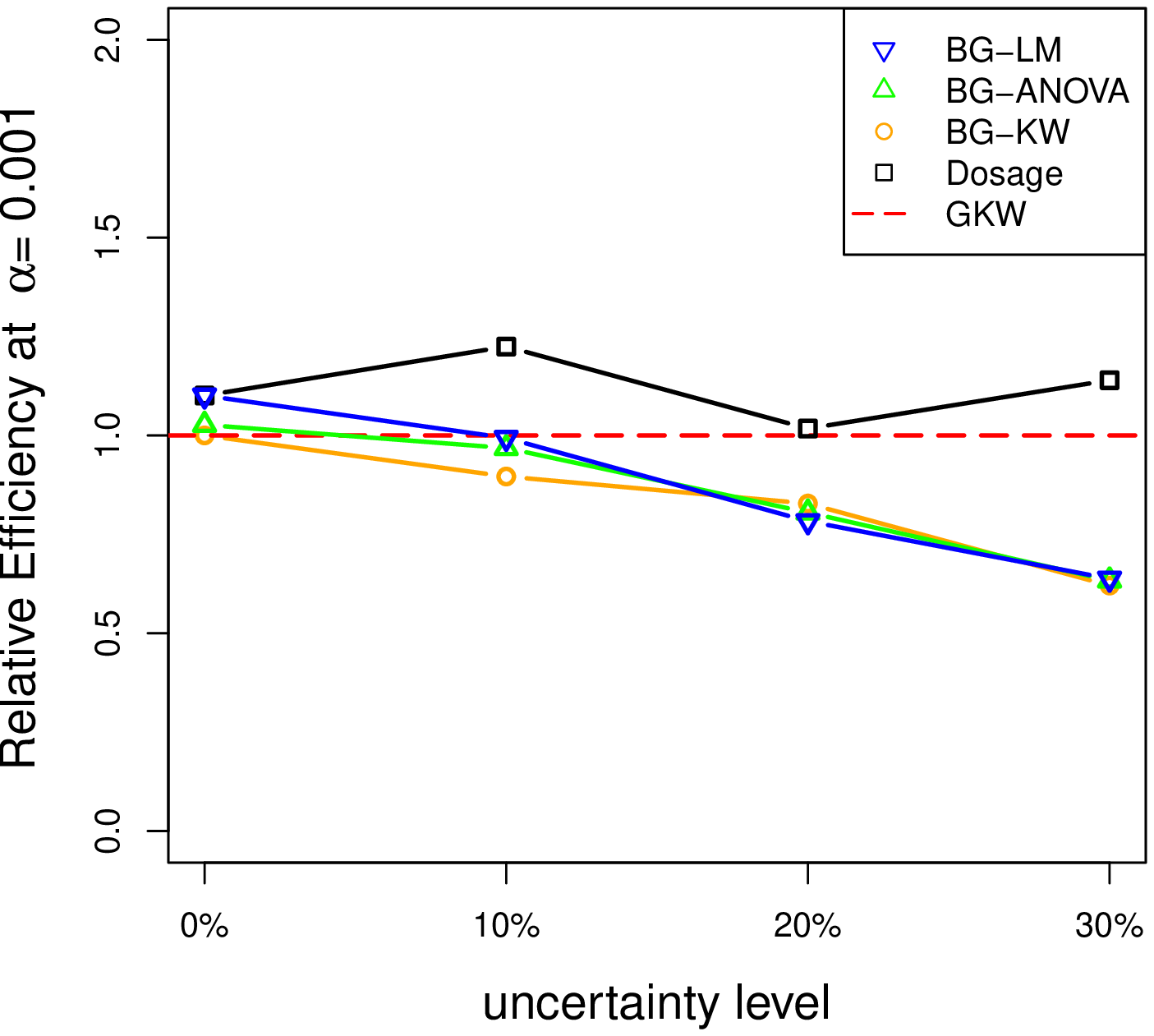} 
}
\caption{Relative efficiency of other tests as compared to the GKW test at (a) $\alpha=0.05$, (b) $\alpha=0.001$, under a normal additive model, for testing the association of SNP that has minor allele frequency of $20\%$.  \label{WebFigure5}}
\end{figure}

\clearpage

\begin{figure}[h!]
\centering
\subfigure[]
{\label{WebFigure6a}
   \includegraphics[width=0.5\textwidth]{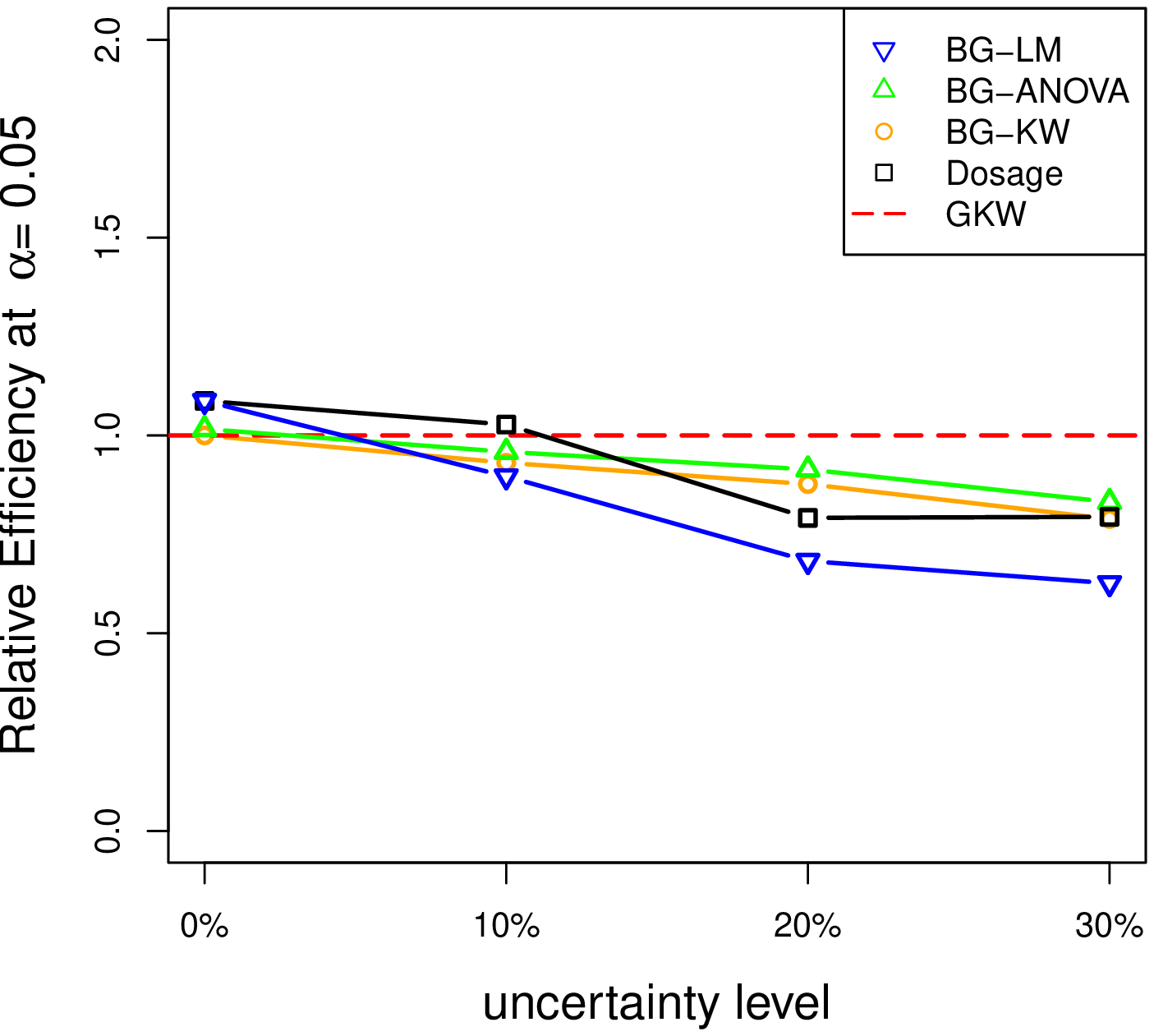}
 }
\hspace{0.5cm}
\subfigure[]
{\label{WebFigure6b}
  \includegraphics[width=0.5\textwidth]{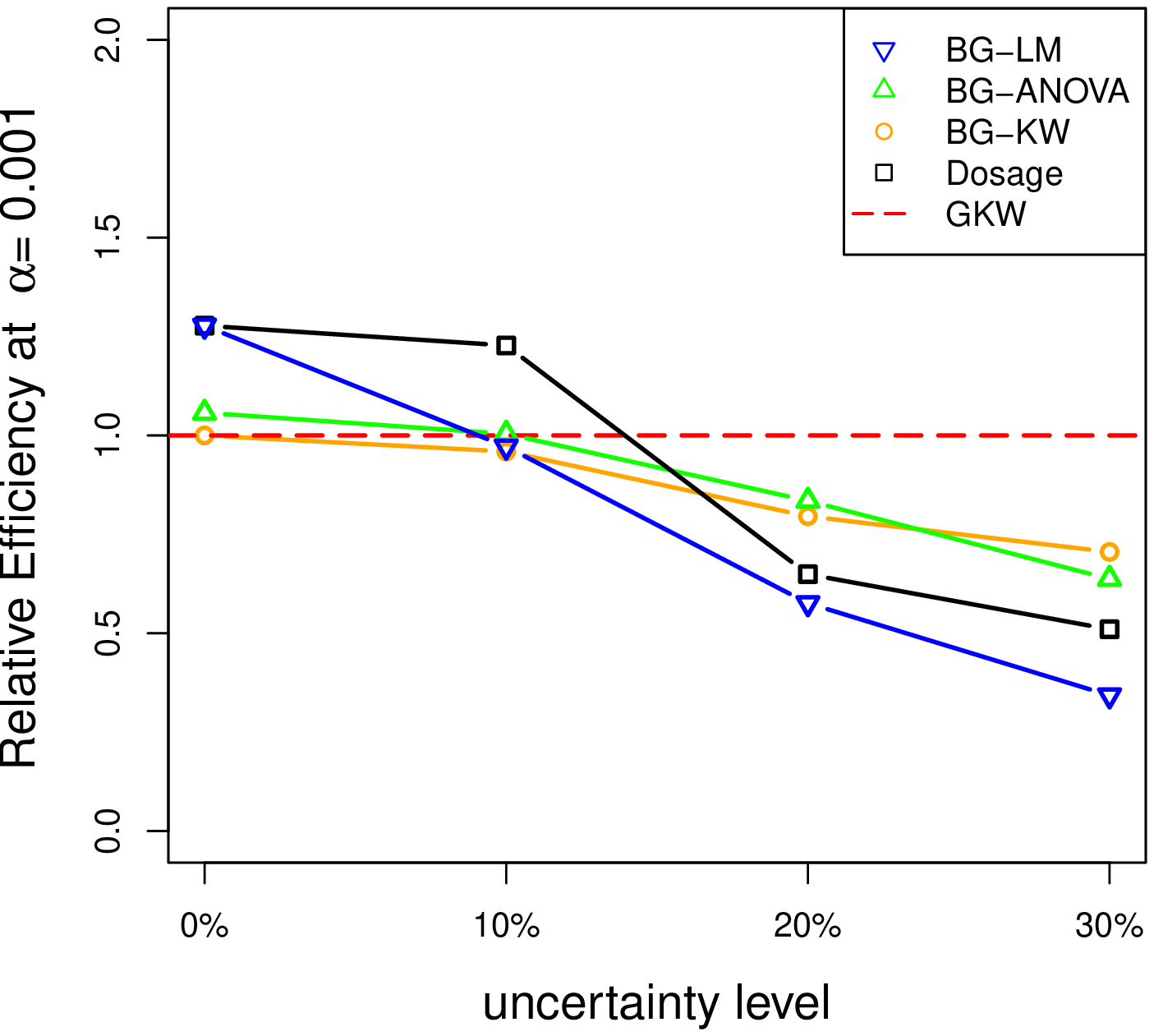} 
}
\caption{Relative efficiency of other tests as compared to the GKW test at (a) $\alpha=0.05$, (b) $\alpha=0.001$, under a normal additive model, for testing the association of SNP that has minor allele frequency of $10\%$. \label{WebFigure6}}
\end{figure}

\clearpage

\begin{figure}[h!]
\centering
\subfigure[]
{\label{WebFigure7a}
\includegraphics[width=0.5\textwidth]{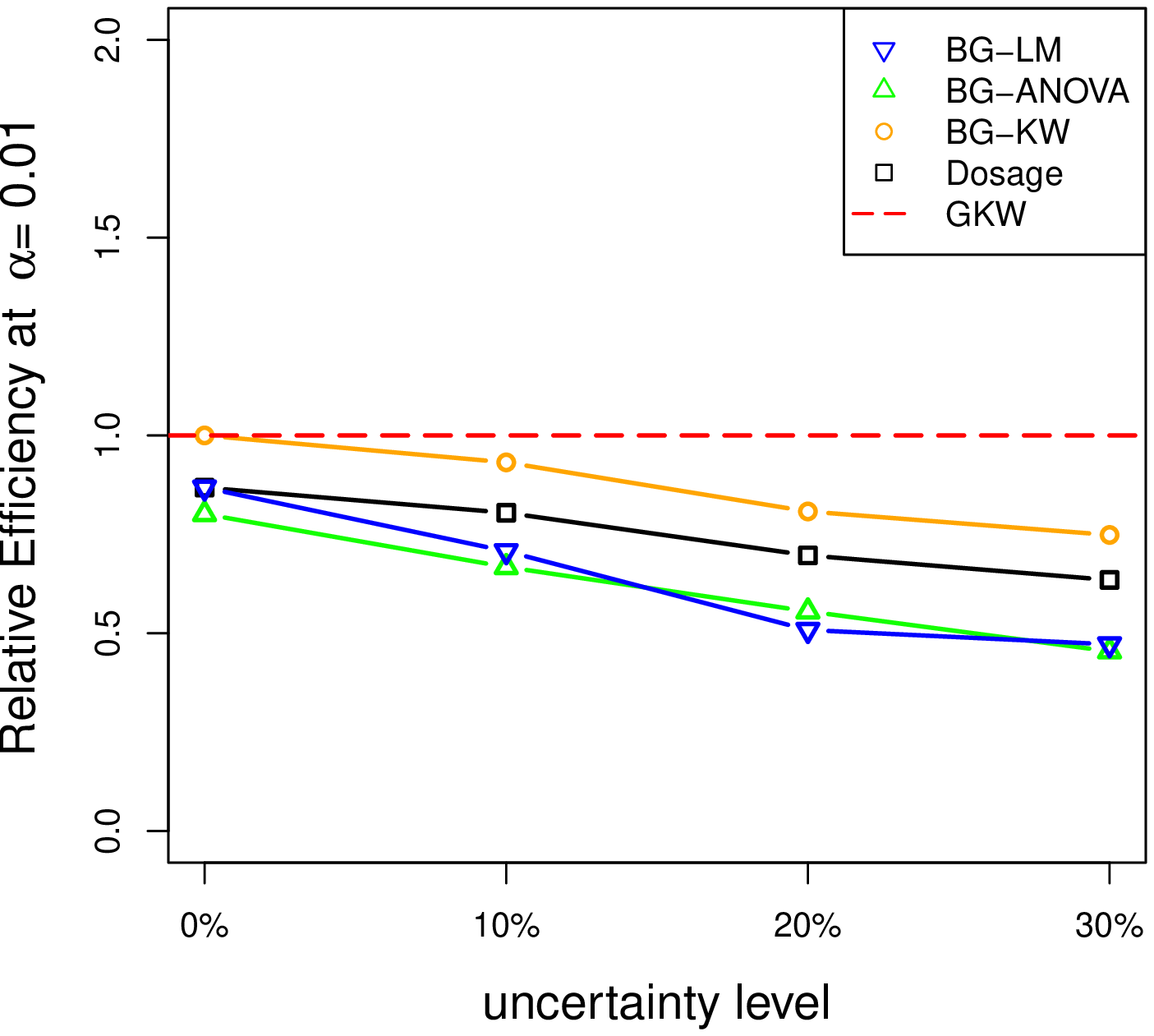}
 }
\hspace{0.5cm}
\subfigure[]
{\label{WebFigure7b}
\includegraphics[width=0.5\textwidth]{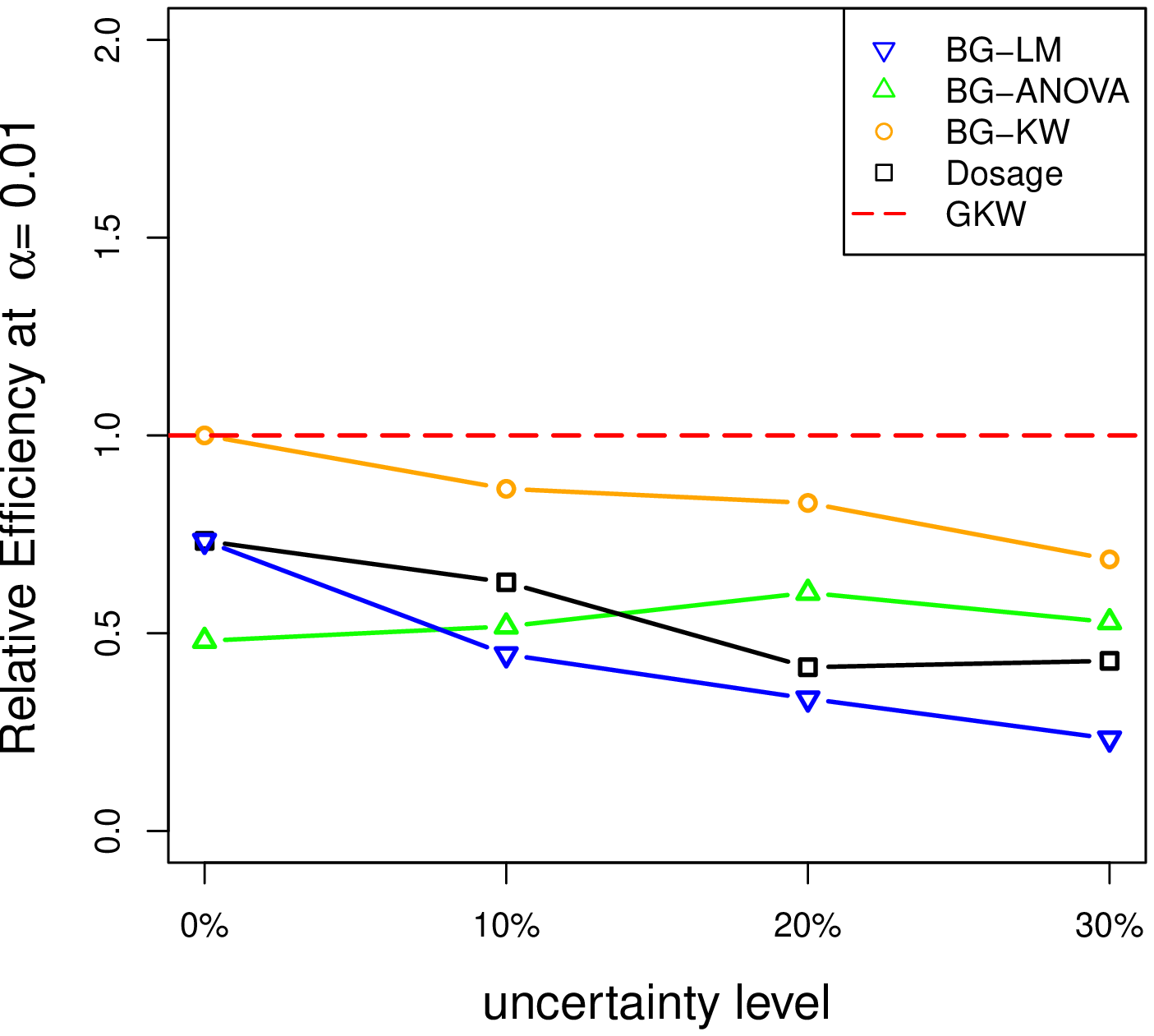}
}
\caption{Relative efficiency of other tests as compared to the GKW test at $\alpha=0.01$ under a non-normal additive model, for testing the association of SNP that has (a) minor allele frequency of $20\%$. (b) minor allele frequency of $10\%$. The non-normal data were obtained by taking the exponent of the normal data generated as in Section 3. \label{WebFigure7}} 
\end{figure}

\begin{figure}[h!]
\centering
\subfigure[]
{\label{WebFigure8a}
\includegraphics[width=0.5\textwidth]{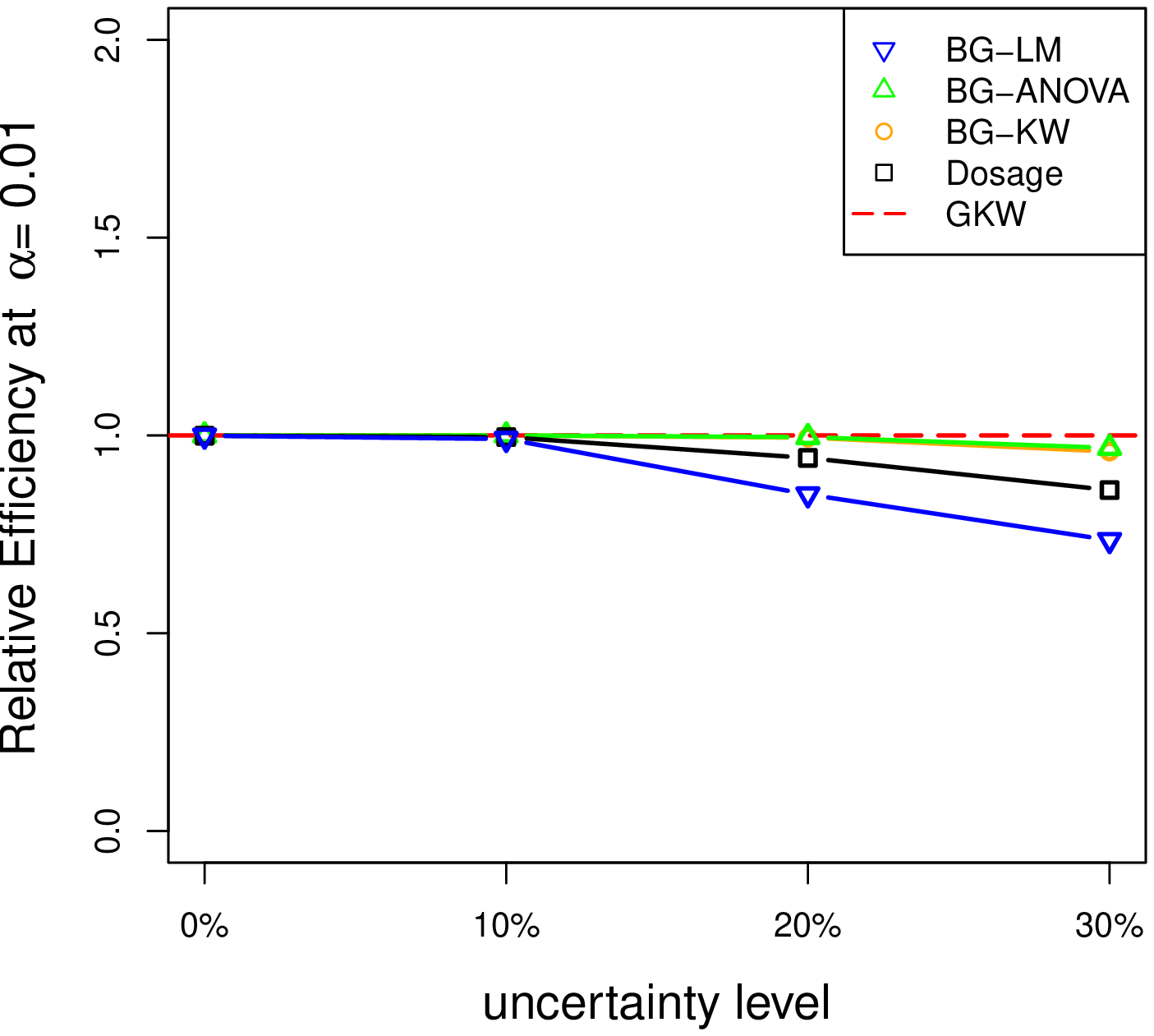}
 }
\subfigure[]
{\label{WebFigure8b}
\includegraphics[width=0.5\textwidth]{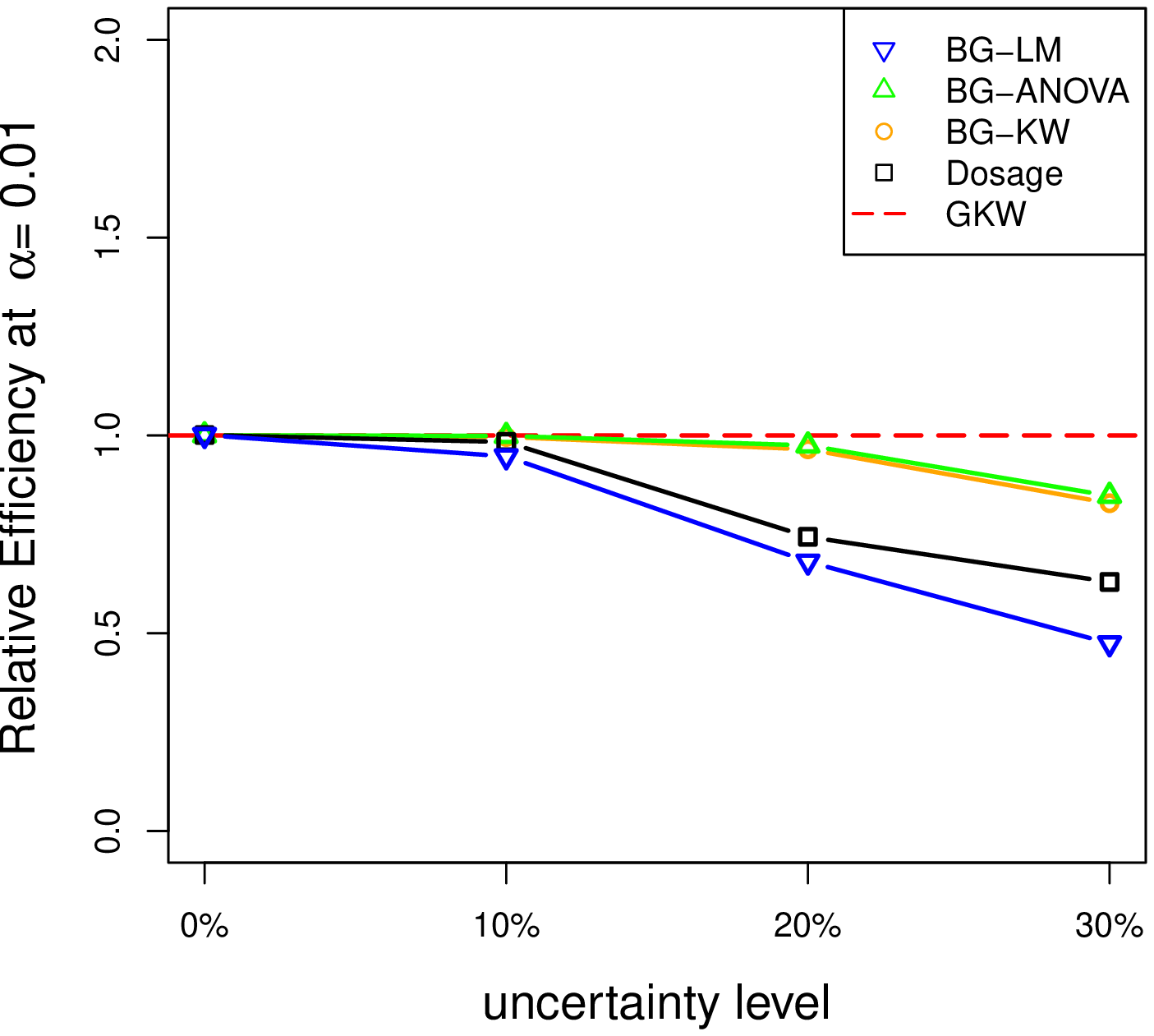}
}
\caption{Relative efficiency of other tests as compared to the GKW test at $\alpha=0.01$ under a normal non-additive model, for testing the association of SNP that has (a) minor allele frequency of $20\%$. (b) minor allele frequency of $10\%$. The data under non-additive model were generated from normal distribution with means $(1.75, 2.25, 2)$ for the three genotype groups $G= 0$, $1$ and $2$, respectively, with a common variance, $\sigma^2=1$. \label{WebFigure8}} 
\end{figure}

\end{document}